\documentclass[conference]{IEEEtran}
\usepackage{fancyhdr}
\usepackage{setspace}
\usepackage{amsmath}
\usepackage{amsfonts}
\usepackage{colortbl}
\usepackage{graphicx}
\usepackage{subfig}
\usepackage{epsfig}
\usepackage{multirow}
\usepackage{physics}
\usepackage{mathrsfs}
\usepackage{amssymb}
\usepackage{amsthm}
\usepackage[noadjust]{cite}
\usepackage{enumerate}
\usepackage{bbm}
\usepackage{microtype}
\usepackage{cleveref}

\allowdisplaybreaks

\def\BibTeX{{\rm B\kern-.05em{\sc i\kern-.025em b}\kern-.08em
T\kern-.1667em\lower.7ex\hbox{E}\kern-.125emX}}

\newtheorem{lemma}{Lemma}

\newtheorem{theorem}{Theorem}

\newcounter{procedure}


\usepackage{algorithmic}
\usepackage{algorithm}

\begin{document}
\title{Delay-Energy Joint Optimization for Task Offloading in Mobile Edge Computing}

\author{
\IEEEauthorblockN{Zhuang Wang, Weifa Liang, Meitian Huang, and Yu Ma}
\IEEEauthorblockA{Research School of Computer Science, The Australian National University, Canberra, ACT 2601, Australia}
}
\maketitle

\begin{abstract}
  Mobile-edge computing (MEC) has been envisioned as a promising paradigm to meet ever-increasing resource demands of mobile users, prolong battery lives of mobile devices, and shorten request response delays experienced by users. An MEC environment consists of many MEC servers and ubiquitous access points interconnected into an edge cloud network. Mobile users can offload their computing-intensive tasks to one or multiple MEC servers for execution to save their batteries. Due to large numbers of MEC servers deployed in MEC, selecting a subset of servers to serve user tasks while satisfying delay requirements of their users is challenging. In this paper, we formulate a novel delay-energy joint optimization problem through jointly considering the CPU-cycle frequency scheduling at mobile devices, server selection to serve user offloading tasks, and task allocations to the selected servers. To this end,
 we first formulate the problem as a mixed-integer nonlinear programming, due to the hardness to solve this nonlinear programming, we instead then relax the problem into a nonlinear programming problem that can be solved in polynomial time. We also show how to derive a feasible solution to the original problem from the solution of this relaxed solution. We finally conduct experiments to evaluate the performance of the proposed algorithm. Experimental results demonstrate that the proposed algorithm is promising.
\end{abstract}


\section{Introduction}\label{sec01}
With the ever-growing popularity of portable mobile devices, including smartphones, tablets, and so on, more and more new mobile applications, such as face recognition and augmented reality, are emerging and have attracted lots of attentions from not only cloud operators but also cellular carriers of the Telecom companies~\cite{SSXP+15, BBCH17}. These emerging applications typically demand intensive computation and high energy consumption. However, mobile devices have limited computation resources and battery life, which become the bottleneck of meeting the quality of experience (QoE) of mobile users.

A traditional approach to overcome the limitation on computation resource and battery life of mobile devices is to leverage cloud computing~\cite{HPSS+15, AA16}, where computation-intensive tasks from mobile devices are offloaded to remote clouds for processing. However, the long latency between mobile devices and remote clouds sometimes is unacceptable for delay-sensitive real-time applications. Mobile-edge computing (MEC) has been introduced in recent years that provides abundant computing resource to mobile users in their close proximity, and mobile users can offload their tasks to one or multiple nearby MEC servers for processing. MEC thus has great potentials to reduce the overall processing delay of computation-intensive tasks and to prolong the battery lifetime of mobile devices. It has attracted lots of attentions from both industries~\cite{HPSS+15,ABBC+17} and academia~\cite{AA16,CSMM+16,LB17}.

A fundamental problem in MEC is the scheduling of offloading tasks, which poses great challenges. For example, due to the limited battery capacity and computation resource, mobile users usually offload their computing-intensive tasks to MEC servers to prolong the battery life of their mobile devices. To meet stringent delay requirement of an offloaded task for delay-sensitive applications such as interactive gaming and augmented reality, each offloaded task needs to be partitioned into multiple chunks and be offloaded to different servers for processing. It becomes critical to optimally allocate tasks from mobile devices to MEC servers in order to minimize the energy consumption of mobile devices while meeting delay requirements of mobile users. Furthermore, when a task is executed on a mobile device, the amount of energy consumed at the mobile device can be minimized by optimally scheduling the CPU-cycle frequency of the mobile device via Dynamic Voltage and Frequency Scaling (DVFS)~\cite{RCN02}. It is worth noting that task allocation needs to be considered with the CPU-cycle frequency scheduling. When performing a task allocation, which size of part of the task should be processed locally, i.e., the number of chunks of the task to be executed on its local mobile device, needs to be determined too. In this paper we will address the mentioned challenges. 
%

The novelty of this paper lies in the formulation of a delay-energy joint optimization problem in MEC and a novel solution to the problem is provided through a series of reductions. 

The main contributions of this paper are summarized as follows.
\begin{itemize}
\item We study task offloading from mobile devices to MEC. We formulate a novel delay {--} energy joint optimization problem for task offloading in MEC that jointly takes into account both response delays and energy consumptions on mobile devices. 

\item We first formulate a mixed-integer nonlinear program for the problem, we then relax the problem into a nonlinear program problem that can be solved in polynomial time. We later show how to derive a feasible solution to the problem from the solution of the relaxed problem.

\item We finally evaluate the performance of the proposed algorithm for the joint optimization problem  through experimental simulation. Experimental results demonstrated the proposed algorithm is very promising, and outperforms its analytical solution as the theoretical estimation is conservative.
\end{itemize}

The rest of this paper is organized as follows. Section~\ref{sec02} introduces the network model and problem formulation. Section~\ref{sec03} provides a solution to  the relaxed version of the problem, and Section~\ref{sec04} shows how to derive a feasible solution to the original problem  for the solution to the relaxed one. Section~\ref{sec05} evaluates the performance of the proposed algorithm through experimental simulations. Section~\ref{sec06} surveys related works, and Section~\ref{sec07} concludes the paper.

\section{Preliminaries}\label{sec02}
In this section, we first introduce the network model. We then introduce the local computing model at each mobile device, including the local execution delay and the energy consumption. We finally introduce the mobile-edge computing model that includes the energy consumption of offloading tasks to MEC through wireless transmission and the processing delay of offloaded tasks in MEC.

\subsection{Network model}
We consider a mobile-edge computing  network (MEC) that consists of  numbers of Access Points (AP) and servers. Denote by $\mathscr{S} = \{s_1, s_2, \cdots, s_N\}$ the set of the servers in the network and $N$ is the number of servers in the MEC, as shown in Fig.~\ref{fig:network_model}. For the sake convenience, we here consider a single mobile device (MD) who accesses a nearby AP to offload his tasks for processing, and we further assume that there is neither wireless channel interferences at the AP nor overloading issues on the servers. We adopt the similar assumptions imposed on each offloading task~\cite{MYZH17}, that is, each task can be arbitrarily divided into different-sized chunks and processed at different servers and/or local mobile devices. We assume that the global information of servers in MEC is given at the AP.

\begin{figure}[hp]
\begin{center}
\includegraphics[scale=0.3]{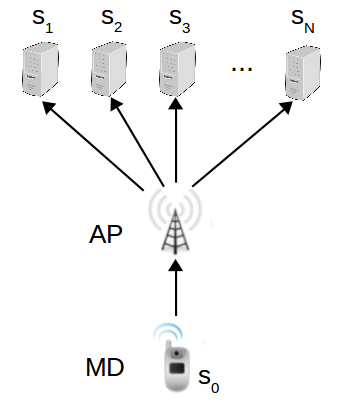}
\end{center}
\caption{An illustrative example of a MEC network with one task offloading from a  mobile user.}
\label{fig:network_model}
\end{figure}

For a given task, whether it is offloaded to MEC entirely or how much proportion of the task will be processed by its mobile device locally, it will determined by a central system scheduler through the AP, which is responsible for the decision and distribution of task chunks to different servers in MEC. 

\subsection{Local computing at the mobile device}
Let $A(L, \tau_d)$ denote a task, where $L$ (in bits) is the input size of the task and $\tau_d$ is the execution deadline, i.e., if the task is to be executed, it should be finished within $\tau_d$ time units. For convenience, denote by $s_0$ the mobile device of the task.  Let $z_i$ be a binary variable, where $z_i= 1$ if the task is partially offloaded to server $s_i$ and $z_i = 0$ otherwise for all $i$ with $ 1\leq i \leq N$, and $z_0=1$ if there is part of the task to be processed at the mobile device. Denote by $\eta_i$ the percentage of task $A(L, \tau_d)$ offloaded to server $s_i$ with $ 0\leq \eta_i \leq 1$. Then, we have
\begin{eqnarray}\label{eqzn}
  \textstyle \sum^N_{i=0} z_i\eta_i = 1,~~\text{$0\leq i \leq N$}.
\end{eqnarray}

A mobile device can execute its task locally. This incurs the local execution delay and energy consumption at the mobile device, while they are jointly determined by the number of CPU cycles of the mobile device.  Denote by $\gamma_A$ the number of CPU cycles required for processing one bit of task $A(L, \tau_d)$, which usually is given and can be obtained through off-line measurements~\cite{MN10}. Hence, the number of CPU cycles needed for processing local chunk of task $A(L, \tau_d)$ is $W=\gamma_Az_0\eta_0L$\footnote{$W$ is truncated to a nearest integer $\lceil W\rceil$ if it is not an integer.}. Since the CPU-cycle frequencies of the mobile device are adjustable through DVFS techniques to reduce the energy consumption, denote by $f_w$ the CPU frequency of the mobile device at the $w$th CPU cycle with $w = 1,\cdots, W$, which is constrained by $f_{\max}$ as follows.
\begin{eqnarray}
    0\leq f_w \le f_{\max}. \label{eqfw}
\end{eqnarray}
As a result, the local execution delay $D_l$ of $A(L, \tau_d)$ is defined as
\begin{eqnarray} \label{eqdl}
    \textstyle D_l = \sum^{W}_{w=1}(f_w)^{-1},
\end{eqnarray}

and the energy consumption $E_l$ of the local execution of $A(L, \tau_d)$ is
\begin{eqnarray}\label{ZZZ}
    \textstyle E_l = \kappa\sum^{W}_{w=1}{f_w}^2,
\end{eqnarray}
where the constant $\kappa$ is the effective switched capacitance that depends on the chip architecture~\cite{MYZH17}.

\subsection{Server computing in MEC}
The mobile device offloads its tasks to MEC via the AP. Such task offloading involves two important metrics: {\it the end-to-end delay} between the time point that a mobile device offloads its task to the MEC and the time point that the execution result of the offloaded task received by the the mobile device; and {\it the energy consumption} of the mobile device on its transmission of the task chunks from the mobile device to the AP. Specifically, the end-to-end delay consists of the transmission delay from the mobile device to the AP, the transmission delay of task chunks of $A(L, \tau_d)$ from the AP to the servers, and the processing delay of task chunks in servers. In the following we give the detailed definition of these two metrics.


Let $r_{h, p}$ denote the maximum achievable uplink transmission rate between the mobile device and the AP which is given in advance. Then, the transmission delay $D_{lp}$ of task chunks of $A(L, \tau_d)$ from the mobile device to the AP is
\begin{eqnarray} \label{eqdlp}
    D_{lp} = \frac{(1-z_0\eta_0)L}{r_{h, p}},
\end{eqnarray}
where $z_0\eta_0$ is the percentage of $A(L, \tau_d)$ processed locally.

The task chunks then are distributed to different servers in MEC once the mobile device finishes its task transmission.
We assume that the number of servers in MEC to which a task can be offloaded is limited by a given integer $m$ ($m \le N$), i.e., the number of servers in MEC to serve a task cannot exceed $m$, because it is inefficient and impractical to offload a task to all servers in a large-scale  MEC.
Since there is no contention for the computational capacity of the servers, the execution of the subtasks can start once the subtasks are offloaded to the servers.

Let $r_i$ denote the transmission rate between the AP and server $s_i$ which is given as well. Then, the  transmission delay $D_{ps}^i$ of the corresponding subtask from the AP to server $s_i$ is
\begin{eqnarray}
    D_{ps}^i = \frac{z_i\eta_iL}{r_i}.
\end{eqnarray}

Let $c_i$ be the computational capability of server $s_i\in {\mathscr{S}}$ (i.e., CPU cycles per second).
The processing delay $D_{sc}^i$ for executing the subtask on server $s_i$ is
\begin{eqnarray}
    D_{sc}^i = \frac{\gamma_Az_i\eta_iL}{c_i}.
\end{eqnarray}

where $\gamma_A$ is the number of CUP cycles required for one-bit data processing.
An offloaded task is finished when all of its subtasks are finished. The end-to-end delay $R_{\max}$ of the offloaded task thus
is defined as the maximum end-to-end delay among all subtasks processed in servers,
while the end-to-end delay of a subtask in server $s_i$  is defined as
the sum of the transmission delays from the mobile device to the AP and from the AP to each chosen server $s_i$
and the task processing delay in $s_i$. Then, $R_{\max}$ is defined as
\begin{eqnarray} \label{eqR}
    R_{\max} = \max_{1\leq i\leq N}\{D_{lp} + D_{ps}^i + D_{sc}^i\}.
\end{eqnarray}

Notice that $R_{\max}$ does not consider the downlink transmission delay, since the downlink data size of $A(L, \tau_d)$ is negligible compared with $L$~\cite{CJLF16}. However, the analysis techniques and the proposed algorithms in this paper are still applicable if the  downlink transmission delay must be considered. 

Denote by $P_{tx}$ the transmission power of the mobile device. We then define the amount of energy consumed $E_{lp}$ by transmitting the task chunks of $A(L, \tau_d)$ from the mobile device to the AP as follows.
\begin{eqnarray} \label{eqelp}
    E_{lp} = \frac{P_{tx}(1-z_0\eta_0)L}{r_{h, p}} + \textbf{1}\{z_0\eta_0 \ne 1\}E_t,
\end{eqnarray}
where 
$\textbf{1}\{x\}$ is an indicator function with 
$\textbf{1}\{x\}$ if $x$ is true and 
$\textbf{1}\{x\} = 0$ otherwise.
$E_t$ is the tail energy due to that the mobile device will continue to hold the channel for a while even after the data transmission~\cite{HC17}.

\subsection{Problem formulation}
%
We define a delay-energy joint optimization problem  $P_0$ for a task offloading of $A(L,\tau_d)$  in MEC with the optimization objective to minimize the weighted sum of the energy consumption at the mobile device and the overall delay of the task execution. 
\begin{eqnarray}
    P_0: & \text{minimize} & E_l + E_{lp} + \alpha \max\{D_l, R_{\max}\} \notag \\
    & \text{s.t.} & (\ref{eqzn})-(\ref{eqelp}), \notag \\
    & & 0 \le \eta_i \le 1, \forall i=0,1,\cdots, N, \label{eqn} \\
    & & \textstyle \sum^{N}_{i=1} z_i \le m, \label{eqzm} \\
    & & \max\{D_l, R_{\max}\} \le \tau_d, \label{eqDRD}
\end{eqnarray}
where $\alpha$ is a coefficient that strives for the tradeoff between the total energy consumption and the overall delay incurred by the task execution~\cite{LCLV17,MA10,LB17}, $z_i$, $\eta_i$ and $f_w$ are variables appeared in Eq.~(\ref{eqzn}) {--} Eq.~(\ref{eqelp}), $E_l + E_{lp}$ is the total amount of energy consumed by the mobile device for task $A(L,\tau_d)$, and $\max \{D_l, R_{\max}\}$ is the overall delay of the task execution which is the maximum of the local execution time at the mobile device and the end-to-end delay in MEC.

Constraint~(\ref{eqn}) indicates that the task can be divided into multiple subtasks and offloaded to different servers. Constraint~(\ref{eqzm}) requires that the number of selected servers is no greater than a given value $m$, while Constraint~(\ref{eqDRD}) is the deadline requirement of task $A(L,\tau_d)$.

\section{A relaxation of the delay-energy joint optimization problem }\label{sec03}
In this section we deal with a relaxation of the problem $P_0$ by providing an optimal solution to the relaxed one as follows.

\subsection{Overview of the proposed algorithm}
 $P_0$ is a mixed integer nonlinear program problem due to  the binary variable $z_i \in \{0,1\}$ with $0\leq i\leq N$ in Constraints~(\ref{eqzn}) and~(\ref{eqzm}) and nonlinear constraints~(\ref{eqdl}) and~(\ref{ZZZ}).

 We  eliminate $z_i$ by assuming that the subset of servers $\mathscr{S}_s$ for task execution while meeting its over delay requirement is given. Then, problem $P_0$ is relaxed to a nonlinear program $P_1$ (Section III.B).


Notice that the optimal solution to $P_1$ later will be used to find an optimal subset of servers for the execution of task $A(L,\tau_d)$ (Section IV.A), and an algorithm is then devised for problem $P_0$ in Section IV.B.

\subsection{Problem relaxation}
$P_0$ includes integer variables $z_i$ and real number variables $\eta_i$ and $f_w$\footnote{The CPU-cycle frequencies are integers. While $f_w$ is a very large integer, often near to ${10}^9$. Hence, we can truncate the real number value of $f_w$ to an integer, which hardly impacts the analysis.}.
We substitute $z_i\eta_i$ with $x_i$ to eliminate integer variables and drop the tail energy $E_t$  in $P_0$.

Denote by $\mathscr{S}_s = \{s_{j_1}, \cdots, s_{j_n}\} \subset \mathscr{S}$ the set of $n$ selected servers for task $A(L,\tau_d)$ while the execution of the offloaded subtasks in MEC meets the overall delay of the task, where $0 < n \le m$, $z_{j_i} = 1$, and $i = 1, \cdots, n$. Hence, Having determined the chosen servers to allocate the subtasks of Task $A(L,\tau_d)$,
problem $P_0$ now can be relaxed into a nonlinear program $P_1$ as follows.
\begin{eqnarray}
    & P_1: \min &  \kappa\sum\limits^{B_0x_0}_{w=1}{f_w}^2 - \phi x_0 + \alpha \max\{D_l, R_{\max}\} \notag \\
    & \text{s.t.} & (\ref{eqfw}), (\ref{eqdl}), (\ref{eqdlp})-(\ref{eqR}),\notag \\
    & & \textstyle \sum^{n}_{i=0} x_i = 1, \label{eqx1} \\
    & & x_i \ge 0, \forall i = 0,1, \cdots,n, \label{eqx2}\\
    & & R_{\max} = \max_{1\leq i \leq n} \{R_i\}\\
    && R_i = q_ix_i + q_0(1-x_0), \forall i=1,\cdots, n,\\
    && \phi  = P_{tx}L/r_{h, p}, \\
    && q_0 = L/r_{h, p},\label{q000} \\
    && q_i = L/r_{j_i} + \gamma_AL/c_{j_i}, \forall i=1,\cdots,n,\label{qi111}\\
    && x_0 = z_0\eta_0, \\
    && x_i = z_{j_i}\eta_{j_i}, \forall i=1,\cdots, n \\
    && B_0 = \gamma_AL, \label{b0}
\end{eqnarray}
where $x_i$ and $f_w$ are variables. Constraints (\ref{eqx1}) and (\ref{eqx2}) in $P_1$ are equivalent to Constraints (\ref{eqzn}), (\ref{eqn}) and (\ref{eqzm}) in $P_0$.
The difference between the objectives of $P_0$ and $P_1$ is $\phi + E_t$. Recall that $z_{j_i}$, $c_{j_i}$, $\gamma_A$, and $r_{h, p}$ are constants.

We assume $\kappa\sum^{B_0x_0}_{w=1}{f_w}^2 > P_{tx}x_0L/r_{h, p}$ in this paper, i.e., the amount of energy consumed for executing a task with size $x_0L$ at the mobile device is greater than that for transmitting the task chunks from the mobile device to the AP without considering the tail energy. Hence, the value of Problem $P_1$ is greater than 0.

\subsection{Solution to the relaxed problem $P_1$}
To solve problem $P_1$, we distinguish it into two cases:
Case 1. $D_l \ge R_{\max}$; 
and Case 2. $D_l \le R_{\max}$. 
We adopt a strategy of solving $P_1$ through the transformation of problem $P_1$ into another nonlinear program without the variable $f_w$, i.e., the CPU-cycle frequency. We then solve the nonlinear program as follows.

{\bf Case 1.}
The local execution delay $D_l$ is no less than the end-to-end delay $R_{\max}$, and thus determines the overall delay. Problem $P_1$ then can be rewritten as the following problem $P_2$.
\begin{eqnarray}
    P_2: & \min & \kappa\sum\limits^{B_0x_0}_{w=1}{f_w}^2 - \phi x_0 + \alpha \sum^{B_0x_0}_{w=1}(f_w)^{-1} \notag \\
    & \text{s.t.} & (\ref{eqfw}), (\ref{eqdlp})-(\ref{eqR}), (\ref{eqx1})~\text{and}~ (\ref{eqx2}) \notag \\
    && \sum^{B_0x_0}_{w=1}{f_w}^{-1} \ge R_{\max},  \label{eqfwR} \\
    & & R_{\max} = \max_{1\leq i \leq n} \{R_i\}. \label{eqRR}
\end{eqnarray}

In the following we solve problem $P_2$ by first solving its simplified version that focuses on CPU frequencies, and then extending this simplified solution to solve $P_2$.

Let $B=B_0x_0$, and let $F(f_1, \cdots, f_B)$ be the sum of the first and third terms in the optimization objective of $P_2$ that includes CPU-frequency variable $f_w$ only, i.e., 
\begin{eqnarray} \label{eqFfw}
    F(f_1,\cdots,f_B) = \kappa\sum^{B}_{w=1}{f_w}^2 + \alpha \sum^{B}_{w=1}(f_w)^{-1}.
\end{eqnarray}

The partial derivative of $F(f_1,\cdots,f_B)$ with respect to $f_w$ is
\begin{eqnarray}
    \pdv{F}{f_w} = 2\kappa f_w - \alpha(f_w)^{-2},
\end{eqnarray}
where $w = 1, \cdots, B$.

 If $(\alpha/2\kappa)^{1/3} \le f_{\max}$, the only stationary point of $F(f_1,\cdots,f_B)$ is $f_w = (\alpha/2\kappa)^{1/3}$ for all $w$ with $1\leq w\leq B$, and
 the minimum value of $F(f_1,\cdots,f_B)$ is $\kappa B (\alpha/2\kappa)^{\frac{2}{3}} + \alpha B (\alpha/2\kappa)^{-\frac{1}{3}}$. Otherwise, we have $\pdv{F}{f_w} < 0$, and  $F(f_1,\cdots,f_B)$ reaches the minimum value of $\kappa B {f_{\max}}^{\frac{2}{3}} + \alpha B {f_{\max}}^{-\frac{1}{3}}$ when $f_w = f_{\max}$ for all $w$
 with $1\leq w\leq B$.

In summary, to achieve the minimum value of $F(f_1,\cdots,f_B)$, each $f_w$ with $1\leq w \leq B$ must be set as follows.
\begin{eqnarray} \label{eqfm}
    f_w = \bar{f} =
    \begin{cases}
        (\frac{\alpha}{2\kappa})^{\frac{1}{3}}, ~~~\text{if $(\frac{\alpha}{2\kappa})^{\frac{1}{3}} \le f_{\max}$} \\
        f_{\max}, ~~~\text{otherwise.}
        \end{cases}
\end{eqnarray}

Having the solution to $F(f_1,\cdots,f_B)$, we now solve problem $P_2$ whose objective is the sum of term $-\phi x_0$ and $F(f_1,\cdots,f_B)$.
Notice that (i) $F(f_1,\cdots,f_B)$ is a function of variables $x_0$ and $f_w$, while item  $-\phi x_0$ is the function of variable $x_0$ only. Thus, the objective function of $P_2$ is the sum of functions  $F(f_1,\cdots,f_B)$ and $-\phi x_0$. Before we solve  $P_2$, we present the following lemma. 

\begin{lemma} \label{lemma333}
    Given three functions $g_1(x,y)$, $g_2(x)$ and $g_3(x,y)=g_1(x,y) + g_2(x)$, where $x$ and $y$ are both sets, then we have that
    \begin{eqnarray}
        \min_{x, y} g_3(x,y) = \min_{x} \{\min_{y} g_1(x,y) + g_2(x)\}. \notag
    \end{eqnarray}
\end{lemma}

\begin{proof}
    The proof can be obtained by contradiction. Let $\{x^*, y^*\}$ be the optimal solution for $\min g_3(x, y)$. Assume that $y^*$ is not the optimal solution for $\min_{y} g_1(x^*, y)$. Then there must exist a $y^{'}$ satisfying
    \begin{eqnarray}
        g_3(x^*, y^*) - g_3(x^*, y^{'}) = g_1(x^*, y^*) - g_1(x^*, y^{'}) > 0, \notag
    \end{eqnarray}
    which contradicts that $g_3(x^*, y^*)$ is the minimum.
    Therefore, the assumption is not true and $y^*$ is the optimal solution for $\min_{y} g_1(x^*, y)$, i.e., the optimal solution for $\min g_3(x, y)$ should be in
    \begin{eqnarray}
        \{(x, y^*)\ |\ y^* = \arg \min_{y} g_1(x, y)\}.  \notag
    \end{eqnarray}
  Lemma~\ref{lemma333} then follows.
\end{proof}

Following Lemma~\ref{lemma333}, we rewrite the optimization objective of $P_2$ by replacing its terms $\kappa\sum^{B}_{w=1}{f_w}^2 + \alpha \sum^{B}_{w=1}(f_w)^{-1}$ with $B(\kappa(\bar{f})^{\frac{2}{3}}+\alpha(\bar{f})^{-\frac{1}{3}})$. We term this transformed optimization problem as problem $P_3$, which is defined as follows.
\begin{eqnarray}
    P_3: & \min &  (\kappa B_0 \bar{f}^2 + \alpha B_0/\bar{f} - \phi)x_0 \notag \\
    & \text{s.t.} & (\ref{eqdlp})-(\ref{eqR}), (\ref{eqx1}), (\ref{eqx2}), (\ref{q000}),  ~\text{and} (\ref{qi111}, ) \notag \\
    && B_0=\gamma_AL\notag\\
    && q_0(1-x_0) + q_ix_i  \le \frac{B_0}{\bar{f}}x_0,
\end{eqnarray}

It can be seen that $P_3$ is a linear program, which can be solved in polynomial time. 

{\bf Case 2.} If the local executing delay $D_l$ is no greater than the end-to-end delay $R_{\max}$, problem  $P_1$ can be rewritten as problem $P_4$ as follows.
\begin{eqnarray}
    P_4: & \min & \kappa \sum\limits^{B_0x_0}_{w=1}{f_w}^2 - \phi x_0 + \alpha R_{\max} \notag \\
    &\text{s.t.} & (\ref{eqfw}), (\ref{eqdlp})-(\ref{eqR}), (\ref{eqx1}), (\ref{eqx2}) ~\text{and } (\ref{eqRR}), \notag \\
    & & \textstyle \sum^{B_0x_0}_{w=1}f_w^{-1} \le R_{\max}. \label{eqfwRRR}
\end{eqnarray}

We adopt the similar strategy for solving $P_4$ as we did for $P_2$. It can be seen the objective function of $P_2$ consists of two functions: function $\kappa \sum\limits^{B_0x_0}_{w=1}{f_w}^2$ has variables $f_w$ and $x_0$; and function
$-\phi x_0 + \alpha R_{\max}$ has variables $x_i$ as $R_{\max}$ is a function of $x_i$ with $0\leq i \leq n$.

By Lemma~\ref{lemma333}, we first determine the value of each $f_w$ to minimize the value of function $\kappa \sum\limits^{B_0x_0}_{w=1}{f_w}^2$ as follows. 

By Inequalities~(\ref{eqfw}) and (\ref{eqfwRRR}), we have 
\begin{eqnarray} \label{eqBfwR}
    \frac{B}{f_{\max}} \le \sum^{B}_{w=1}\frac{1}{f_w} \le R_{\max}.
\end{eqnarray}
By Inequality (\ref{eqBfwR}), we have 
\begin{eqnarray}
    B /R_{\max} \le f_{\max}.
\end{eqnarray}
Denote by $R'=\sum\limits^{B}_{w=1}({1}/{f_w})$. \\
Then, $R' \le R_{\max}$ by Inequality~(\ref{eqfwRRR}).

By Jensen's inequality (see the appendix), we then have 
\begin{eqnarray} \label{eqjensen}
    \sum^{B}_{w=1}{f_w}^2 \ge \frac{B^3}{{R^{'}}^2} \ge \frac{B^3}{R_{\max}^2}.
\end{eqnarray}
Hence, $\sum^{B}_{w=1}{f_w}^2$ reaches the minimum ${B^3}/{R_{\max}^2}$ when \break $f_w = \frac{B}{R_{\max}}$ for all $w$ with $1 \le w \le B$. Then,
$D_l =  R_{\max}= \sum^{B}_{w=1}(\frac{1}{f_w})$.

Problem $P_5$ is then derived as follows, by replacing $\sum^{B}_{w=1}{f_w}^2$ in $P_4$ with $B^3/R_{\max}^2$.
\begin{eqnarray}
    P_5: & \min & \frac{Kx_0^3}{R_{\max}^2} - \phi x_0 + \alpha R_{\max} \notag \\
    & \text{s.t.} & (\ref{eqdlp})-(\ref{eqR}), (\ref{eqx1}), (\ref{eqx2}), (\ref{b0}) ~\text{and}~ (\ref{eqRR}), \notag\\
& & K=\kappa B_0^3.
\end{eqnarray}


\begin{lemma} \label{lemmaP1P5}
    $P_1$ is equivalent to $P_5$.
\end{lemma}

\begin{proof}
It can be seen that problem $P_3$ is a special case of problem $P_5$ when $R_{\max} = (B_0x_0)/{f_m}$. Under Case 1, $P_1$ is equivalent to $P_3$; while under Case 2, $P_1$ is equivalent to $P_5$. Hence, $P_1$ is equivalent to $P_5$.
\end{proof}

The rest is to solve the nonlinear program problem {--} $P_5$.
Before we proceed, we have the following lemma. 
\begin{lemma} \label{lemma12}
    If $f(\cdot)$ is a real-valued function of two variables of $x$ and $y$, which is defined whenever $x \in X$ and $y \in Y$ and $X$ and $Y$ are two domains. Then
    \begin{eqnarray}\label{claim000}
        \min_{x\in X, y\in Y} f(x, y) = \min_{x \in X} {\min_{y \in Y} f(x, y)}.
    \end{eqnarray}
\end{lemma}
\begin{proof}
    We show the claim by contradiction. It can be seen that
    \begin{eqnarray}\label{xy001}
        \min_{x\in X, y\in Y} f(x, y) \le \min_{x \in X} {\min_{y \in Y} f(x, y)}. \notag
    \end{eqnarray}
    
    Assuming that $\{x^*, y^*\}$ is the optimal solution of $ \min_{x\in X, y\in Y} f(x, y)$. That is, 

    \begin{eqnarray}\label{aa11}
      f(x^*, y^*) = \min_{x \in X} {\min_{y \in Y} f(x, y)}.
    \end{eqnarray}
    
    Suppose the claim~(\ref{claim000}) does not hold, and we assume that 
    \begin{eqnarray}\label{bb11}
        \min_{x\in X, y\in Y} f(x, y) < \min_{x \in X} {\min_{y \in Y} f(x, y)}. \notag
    \end{eqnarray}

    By Ineq.~(\ref{bb11}),  we have 
    \begin{eqnarray}
        f(x^*, y^*) < \min_{y \in Y} f(x^*, y). \notag
    \end{eqnarray}
    However, by Eq.~(\ref{aa11}), we have 
    \begin{eqnarray}
        f(x^*, y^*) \ge \min_{y \in Y} f(x^*, y). \notag
    \end{eqnarray}
    This leads to a contradiction. The lemma thus follows.
\end{proof}

It can be seen that the objective function of $P_5$ contains variables  $x_0$ and $R_{\max}$. Denote by $H(x_0, R_{\max})$ the objective function of $P_5$ which is defined as follows. 
\begin{eqnarray}
    H(x_0, R_{\max}) = \frac{Kx_0^3}{R_{\max}^2}  - \phi x_0 + \alpha R_{\max}.
\end{eqnarray}

To minimize $H(x_0, R_{\max})$, we first determine the range of $R_{\max}$ with respect to $x_0$ in order to meet Eq.~(\ref{eqRR}). We then transform problem $P_5$ into a problem with respect to $x_0$ only, by adopting  Lemma~\ref{lemma12}. The range of $R_{\max}$ is determined by the following lemma.
\begin{lemma} \label{lemmaintervalR}
The range of $R_{\max}$ with respect to $x_0$ is
\begin{eqnarray}
        (1-x_0)\bar{Q} \le R_{\max} \le (1-x_0)Q_u \label{eqMRM}
    \end{eqnarray}
    where
    \begin{align}
        & Q = \sum^{n}_{i=1} {1}/{q_i}\\
        & \bar{Q} = q_0 + 1 / Q, \label{eqM}\\
        & Q_u = q_0 + \max_{1\leq i\leq n}\{q_i\}.
    \end{align}
    $R_{\max}$ reaches the minimum $\bar{Q}(1-x_0)$ when
    \begin{eqnarray}
        x_i = \frac{1-x_0}{Q\cdot q_i} ~~~\text{for all $i$ with $1\leq i \leq n$}.\label{eqxiQ0}
    \end{eqnarray}
\end{lemma}

\begin{proof}
    Recall that $R_{\max} = \max_{1\leq i \leq n} \{R_i\}$ by Eq.~(\ref{eqRR}). Determining the range of $R_{\max}$ is equivalent to determining the upper and lower bounds of $\max_{1\leq i\leq n}\{R_i\}$ ($=\max_{1\leq i\leq n}\{q_ix_i + q_0(1-x_0)\}$), subject to that $\sum^{n}_{i=1}x_i = 1-x_0$ with $x_i \ge 0$.

    We first derive the upper bound of $R_{\max}$.
    As $x_i \le 1 - x_0$ and $q_i \le \max_{1 \le i \le n} \{q_i\}$, we have
        $\max_{1\le i \le n}\{q_ix_i\} \le \max_{1 \le i \le n}\{q_i\}(1-x_0)$.
    Hence,  $R_{\max} \le (1-x_0)Q_u$ holds.

We then derive the lower bound of $R_{\max}$, i.e., $\min \max_{1 \le i \le n}\{q_ix_i + q_0(1-x_0)\}$.  As term $q_0(1-x_0)$ is independent of $R_{\max}$, we will focus on term $q_ix_i$. In the following we show that $\min \max_{1 \le i \le n} \{q_ix_i\}$ reaches the minimum $(1-x_0)/Q$ when $x_i=\frac{1-x_0}{Q\cdot q_i}$ by contradiction.

Without loss of generality, we suppose that the optimal solution of $\min \max_{1\leq i \leq n} \{q_ix_i\}$ satisfies that $q_1x_1 \le q_2x_2 \le \cdots \le q_nx_n$. Then, $q_nx_n - q_1x_1 = \epsilon \geq 0$. We now show that $q_nx_n - q_1x_1 = 0$ by contradiction. Assume that $q_nx_n - q_1x_1 = \epsilon> 0$,  we then can have another feasible solution for $\min \max_{1\leq i \leq n} \{q_ix_i\}$ \break $\mathscr{X}^{'} = \{x_1^{'}, x_2', \cdots, x_{n-1}', x_n^{'}\}$, where
    \begin{eqnarray}
        x_1^{'} &=& x_1 + \frac{\epsilon}{x_1+x_n}, \notag \\
        x_i^{'}    &=& x_i, ~~~2 \le i \le n-1, \notag \\
        x_n^{'} &=& x_n - \frac{\epsilon}{x_1+x_n}. \notag
    \end{eqnarray}
    Since $q_nx_n^{'} < q_nx_n$, the value of $\min \max_{1 \le i \le n} \{q_ix_i\}$ under $\mathscr{X}^{'}$ is $\min\{q_{n-1}x_{n-1}, q_nx_n^{'}\} < q_nx_n$. $\mathscr{X}^{'}$ is better than $\mathscr{X}$. This leads to a contradiction that $\mathscr{X}$ is the optimal solution to the problem.  Therefore,
    \begin{eqnarray}
        q_1x_1 = q_2x_2 = \cdots = q_nx_n. \notag
    \end{eqnarray}
    That is, $\min \max_{1 \le i \le n} \{q_ix_i\}$ reaches the minimum $(1-x_0)/Q$
    when Eq.~(\ref{eqxiQ0}) holds. We thus have $R_{\max} \ge (1-x_0)/Q + q_0(1-x_0) = (1-x_0)\bar{Q}$, thereby $(1-x_0)\bar{Q} \le R_{\max}$. Inequalities~\ref{eqMRM} thus hold.     
\end{proof}

Since $R_{\max}$ is a function of $x_0$, $\min_{R_{\max}} H(x_0, R_{\max})$ is also a function of $x_0$ too.
\begin{lemma} \label{lemmaRequivalent}
    $\min_{R_{\max}}H(x_0, R_{\max})$ is equivalent to function $h(x)$, where $h(x)$ is defined as follows.
    \begin{eqnarray}
        h(x) =
        \begin{cases}
            h_1(x), x \in [0, \frac{\bar{Q}}{R^*+\bar{Q}}] \\
            h_2(x), x \in [\frac{Q_u}{R^*+Q_u}, 1] \\
            h_3(x), x \in [\frac{\bar{Q}}{R^*+\bar{Q}}, \frac{Q_u}{R^*+Q_u}]
        \end{cases}
    \end{eqnarray}
    where
    \begin{eqnarray}
        h_1(x) &=& \frac{K}{\bar{Q}^2}\frac{x^3}{(1-x)^2} - \phi x + \alpha (1-x)\bar{Q}, \\
        h_2(x) &=& \frac{K}{Q_u^2}\frac{x^3}{(1-x)^2} - \phi x + \alpha (1-x)Q_u, \\
        h_3(x) &=& (\frac{K}{{R^*}^2} - \phi + \alpha R^*)x, \label{eqh3} \\
        R^* &=& (\frac{2K}{\alpha})^{\frac{1}{3}}.
\end{eqnarray}
\end{lemma}

\begin{proof}
    The partial derivative of $H(x_0, R_{\max})$ with respect to $R_{\max}$ is
        $\pdv{H(x_0, R_{\max})}{R_{\max}} = -2K\frac{x_0^3}{R_{\max}^3} + \alpha$.
    Let $\pdv{H(x_0, R_{\max})}{R_{\max}} = 0$, its solution is $R^*x_0$.
    If $R^*x_0 \ge (1-x_0)Q_u$, i.e., $\frac{Q_u}{R^*+Q_u} \le x_0 \le 1$, $H(x_0, R_{\max})$ is a monotonically decreasing function.
    Hence, we have $\min_{R_{\max}}H(x_0, R_{\max}) = H(x_0, (1-x_0)Q_u)$, which can be rewritten as $h_2(x)$.
    Similarly, we can obtain the results for $R^*x_0 \le (1-x_0)\bar{Q}$ and $(1-x_0)\bar{Q} \le R^*x_0 \le (1-x_0)Q_u$.
\end{proof}

Problem $P_5$ can be solved by finding the minimum value of $h(x)$, while the minimum value of $h(x)$ is the minimum one among $h_1(x)$, $h_2(x)$ and $h_3(x)$ that are defined as follows. For convenience, denote $\min h_i(x)$ with $1 \le i \le 3$ the minimum value of $h_i(x)$ in its domain.
\begin{lemma} \label{lemmaminh}
    The minimum values of $h_1(x), h_2(x)$ and $h_3(x)$ respectively are
    \begin{align}
        & \min_{x \in [0, \frac{\bar{Q}}{R^*+\bar{Q}}]} \{h_1(x) \}= h_1(x_0^*), \\
        & \min_{x \in [\frac{Q_u}{R^*+Q_u}, 1]} \{h_2(x)\} = h_2(\frac{R^*Q_u}{R^*+Q_u}), \\
        &\min_{x \in [\frac{\bar{Q}}{R^*+\bar{Q}}, \frac{Q_u}{R^*+Q_u}]} \{h_3(x)\}= h_3(\frac{R^*\bar{Q}}{R^*+\bar{Q}}),  \label{eqminh3}
    \end{align}
    where $x_0^*$ satisfies
    \begin{align}
       & y^* = \frac{x_0^*}{1-x_0^*}, \label{eqyx} \\
        & 2{y^*}^3 + 3{y^*}^2 = \frac{(\phi+\alpha \bar{Q})\bar{Q}^2}{K}. \label{eqyM}
    \end{align}
\end{lemma}
\begin{proof}
    Since function $h_3(x)>0$ defined in Eq.~(\ref{eqh3}), we have that $\frac{K}{{R^*}^2} + \alpha R^* > \phi$. Hence, 
    $\min h_3(x) = h_3(\frac{\bar{Q}}{R^*+\bar{Q}})$.

    We now find the minimum value of $h_1(x)$. The derivative of $h_1(x)$ is 
    \begin{eqnarray}
        \pdv{h_1(x)}{x} = \frac{K}{\bar{Q}^2}(2y^3 + 3y^2) - (\phi + \alpha \bar{Q}), \notag
    \end{eqnarray}
    where $y = x/(1-x)$.
    
    We define $h_4(y)$ as follows.
    \begin{eqnarray}
        h_4(y) = 2y^3 + 3y^2 - \frac{(\phi+\alpha \bar{Q})\bar{Q}^2}{K}, 0 \le y \le \frac{\bar{Q}}{R^*}. \notag
    \end{eqnarray}
    As $h_4(y)$ is a strictly monotonically increasing function with respect to $y$, we have $\min \{h_4(y)\} = h_4(0) < 0$. The rest is to show that $\max \{h_4(y)\} = h_4(\frac{\bar{Q}}{R^*}) > 0$.
    \begin{align*}
            & h_4(\frac{\bar{Q}}{R^*}) = 2(\frac{\bar{Q}}{R^*})^3 + 3(\frac{\bar{Q}}{R^*})^2 -\frac{(\phi+\alpha \bar{Q})\bar{Q}^2}{K} \\
                                 &> 2(\frac{\bar{Q}}{R^*})^3 + 3(\frac{\bar{Q}}{R^*})^2 - \frac{(K/{R^*}^2 + \alpha R^* + \alpha \bar{Q})\bar{Q}^2}{K} \\
                                 &= 0.
    \end{align*}
    Therefore, $h_4(y) = 0$ has the only solution $y^*$ satisfying Eq.~(\ref{eqyM}). Since $y$ is a monotonically increasing function with respect to $x$, $\pdv{h_1(x)}{x} = 0$ also has the only solution $x_0^*$.
    Therefore, $h_1(x)$ reaches its minimum $h_1(x_0^*)$ when Eq.~(\ref{eqyx}) and Eq.~(\ref{eqyM}) hold.
    Similarly, $h_2(x)$ is a monotonically increasing function of $x$ too, and reaches its minimum when $x=\frac{Q_u}{R^*+Q_u}$.
\end{proof}

\begin{theorem} \label{theoremOPTP1}
    The optimal solution to problem $P_1$ is $OPT(P_1^{\mathscr{S}_s})$ when  
    \begin{eqnarray} \label{eqOPTP1}
    OPT(P_1^{\mathscr{S}_s}) = 3K \frac{{y^*}^2}{\bar{Q}^2} - \phi,
    \end{eqnarray}
    where 
    \begin{eqnarray}
        x_0 &=& x_0^*, \\
        x_i &=& \frac{1-x_0^*}{Q\cdot q_i}, \label{eqxiQ}
    \end{eqnarray}

    and the  CPU-cycle frequency scheduling at the mobile device is
    \begin{eqnarray}
        f_w = \frac{B_0x_0^*}{(1-x_0^*)\bar{Q}},~~~\text{for all $w$ with $1\leq w \leq B_0x_0^*$}. \label{eqfwM}
    \end{eqnarray}
    and the corresponding overall delay is $(1-x_0^*)\bar{Q}$.
\end{theorem}

\begin{proof}
    According to Lemmas~\ref{lemmaRequivalent} and~\ref{lemmaminh}, we have
    \begin{eqnarray}
        \min_{x_0, R_{\max}} H(x_0, R_{\max}) = \min \{h_1(x), h_2(x), h_3(x) \}. \notag
    \end{eqnarray}
    Since $x_0^* \in [0, \frac{\bar{Q}}{R^* + \bar{Q}}]$, we have
        $\min h_1(x) = h_1(x_0^*) \le h_1(\frac{R^*\bar{Q}}{R^*+\bar{Q}})$.
    As $h_3(x)$ is a monotonically increasing function, we have 
    $\min h_3(x)= h_3(\frac{R^*\bar{Q}}{R^*+\bar{Q}})$.

    Notice that $h_1(x)$ and $h_3(x)$ represent the same point when $x=\frac{R^*\bar{Q}}{R^*+\bar{Q}}$, thus we have $h_1(\frac{R^*\bar{Q}}{R^*+\bar{Q}}) = h_3(\frac{R^*\bar{Q}}{R^*+\bar{Q}})$.
    Therefore, $ \min h_1(x) \le \min h_3(x)$ holds.
    Similarly,  
    $\min {h_3(x)} \le \min {h_2(x)}$. 
    Therefore,
    we have $\min h_1(x) \le \min h_3(x) \le \min h_2(x)$ and
    \begin{eqnarray*}
        \min_{x_0, R_{\max}} H(x_0, R_{\max}) = h_1(x_0^*).
    \end{eqnarray*}
    Thus, the optimal solution to problem $P_5$ is $h_1(x_0^*)$, which is also the optimal solution to $P_1$ as $P_1$ to $P_5$ is equivalent by Lemma~\ref{lemmaP1P5}.
    By Lemma~\ref{lemmaintervalR}, $R_{\max}=(1-x_0^*)\bar{Q}$ when $x_i$ satisfies Eq.~(\ref{eqxiQ}). Meanwhile, $R_{\max} = R_i$ with $1 \le i \le n$. By the analysis of Case 2. in Section III.C, $R_{\max} = D_l$ and $f_w = B_0x_0^*/R_{\max}$. The overall delay $\max\{D_l, R_{\max}\} = (1-x_0^*)\bar{Q}$.

    We now rewrite $OPT(P_1^{\mathscr{S}_s})$ as a function of $\bar{Q}$. Eq.~(\ref{eqyx}) can be rewritten as
    \begin{eqnarray}
    x_0^* = \frac{y^*}{1+y^*}. \label{eqxy}
    \end{eqnarray}
    Hence, $OPT(P_1^{\mathscr{S}_s})$ can be simplified as
    \begin{eqnarray} \small
       OPT(P_1^{\mathscr{S}_s}) &=& \frac{K}{\bar{Q}^2}\frac{{x_0^*}^3}{(1-x_0^*)^2} -\phi x_0^* + \alpha \bar{Q}(1-x_0^*) \notag  \\
       &=& \frac{K}{\bar{Q}^2}\frac{{y^*}^3}{1+y^*} - \frac{\phi y^*}{1+y^*} + \frac{\alpha \bar{Q}}{1+y^*} \notag \\
       &=& \frac{K{y^*}^3-\phi \bar{Q}^2y^* + \alpha \bar{Q}^3}{\bar{Q}^2(1+y^*)} \label{eqaMy} \\
       &=& \frac{K{y^*}^3-\phi \bar{Q}^2y^* + (2K{y^*}^3+3K{y^*}^2-\phi \bar{Q}^2)}{\bar{Q}^2(1+y^*)} \notag \\
       &=& \frac{3K{y^*}^2(1+y^*)-\phi \bar{Q}^2(1+y^*)}{\bar{Q}^2(1+y_n)} \notag \\
       &=& 3K \frac{{y^*}^2}{\bar{Q}^2} - \phi. \notag
    \end{eqnarray}
    where term $\alpha \bar{Q}^3$ in Eq.~(\ref{eqaMy}) is replaced by $2K{y^*}^3+3K{y^*}^2-\phi \bar{Q}^2$, following Eq.~(\ref{eqyM}).
\end{proof}

\section{Algorithm for the delay-energy joint optimization problem}\label{sec04}
In this section we solve the original problem $P_0$, by making use of the result of problem $P_1$. The relationship between $P_0$ and $P_1$ is given as follows. Let $OPT(P_0)$ and $OPT(P_1^{\mathscr{S}_s})$ be the optimal solutions of $P_0$ and $P_1$ with the given server set $\mathscr{S}_s$, respectively. Notice that the set of servers $\mathscr{S}_s$ allocated to for the subtask executions of task $A(L,\tau_d)$ in problem $P_1$ is given, under which an optimal solution for problem $P_1$ is achieved. In other words, given a different server set $\mathscr{S}_s'$ for $P_1$, there is a different optimal solution $OPT(P_1^{\mathscr{S}_s'})$ for it.

The general strategy for problem $P_0$ is to reduce it to $P_1$ by fixing the server set for each offloading task $A(L,\tau_d)$. Since there are multiple subsets of servers satisfying the overall delay $\tau_d$ of task $A(L,\tau_d)$, the rest of this section will focus on identifying such a subset of servers that the optimal solution of $P_1$ is the minimum one, which corresponds an optimal solution to problem $P_0$. Denote by $OPT(P_1^{\mathscr{S}_{s_0}})$ the minimum one among $OPT(P_1^{\mathscr{S}_{s_0}})$ for all $\mathscr{S}_s$, i.e.,  
\begin{eqnarray}
OPT(P_1^{\mathscr{S}_{s_0}})=\min\{OPT(P_1^{\mathscr{S}_s}) | \mathscr{S}_s \subseteq \mathscr{S}, \ |\mathscr{S}_s|\leq m\}
\end{eqnarray}

Let $OPT_{local}$ be the minimum value of the overall cost if task $A(L,\tau_d)$ is entirely executed at the mobile device. 
Then, $OPT(P_0)$ is the minimum between $OPT_{local}$ and $OPT(P_1^{\mathscr{S}_{s_0}}) + \phi + E_t$, which is the minimum overall cost if the task is offloaded to MEC. The rest of this section thus focus on identifying the server set $s_0$.  

\subsection{Server selection}

We first derive the relationships between $OPT(P_1^{\mathscr{S}_s})$ and $\bar{Q}$ defined in ~(\ref{eqM}) and between $R_{\max}$ and $\bar{Q}$. We then find the server subset $s_0$ in the MEC to serve task $A(L, \tau_d)$ based on these relationships.
\begin{lemma} \label{lemmaOPTP1M}
   $OPT(P_1^{\mathscr{S}_s})$ is a monotonically increasing function with respect to $\bar{Q}$.
\end{lemma}

\begin{proof}
   Denote by $\xi = y^*/\bar{Q}$ a function of $\bar{Q}$, clearly $\xi > 0$. $OPT(P_1^{\mathscr{S}_s})$ in Eq.~(\ref{eqOPTP1}) then can be rewritten as $OPT(P_1^{\mathscr{S}_s}) = 3K\xi^2 - \phi>0$, and Eq.~(\ref{eqyM}) can be similarly rewritten as
        $2K\frac{{y^*}^3}{\bar{Q}^3} + 3K\frac{{y^*}^2}{\bar{Q}^3} = \frac{\phi}{\bar{Q}} + \alpha$,
    which can be further simplified as
        $2K\bar{Q}\xi^3 + 3K \xi^2 = \phi + \alpha \bar{Q}$.
    Hence, $\bar{Q}$ can be expressed as
    \begin{eqnarray} \label{eqMd}
        \bar{Q} = \frac{3K\xi^2-\phi}{\alpha - 2K\xi^3}.
    \end{eqnarray}
   Since $3K\xi^2 - \phi > 0$, $\alpha - 2K\xi^3>0$. Thus, $\bar{Q}$ is a strictly monotonically increasing function of $\xi$. With the property of the inverse function, $\xi$ is also a strictly monotonically increasing function of $\bar{Q}$.
    Since $\xi > 0$, $OPT(P_1^{\mathscr{S}_s})$ is a monotonically increasing function of $\xi$. Hence, $OPT(P_1^{\mathscr{S}_s})$ is a monotonically increasing function of $\bar{Q}$.
\end{proof}

We now show that decreasing the value of $\bar{Q}$ reduces the overall delay of the execution of task $A(L,\tau_d)$ by the following lemma.

\begin{lemma} \label{lemmaRM}
    The overall delay of task $A(L, \tau_d)$ is a monotonically increasing function with respect to $\bar{Q}$.
\end{lemma}

\begin{proof}
   By replacing $x_0$ in Eq.~(\ref{eqxy}) with $x_0=1- R_{\max}/\bar{Q}$ by Theorem~\ref{theoremOPTP1},  $y^*$ can be expressed as follows.  
\begin{eqnarray}
        y^* = \frac{\bar{Q}}{R_{\max}} - 1. \label{eqyMR}
\end{eqnarray}
Eq.~(\ref{eqyM}) can be rewritten by replacing $y^*$ with $\frac{\bar{Q}}{R_{\max}} - 1$ as follows.
\begin{eqnarray}
    2\bar{Q}^3 - 3 \bar{Q}^2R_{\max} + R_{\max}^3 = (\frac{\phi}{K}\bar{Q}^2 + \frac{\alpha}{K}\bar{Q}^3)R_{\max}^3. \label{eqMR}
\end{eqnarray}

   The differentiation of Eq.~(\ref{eqMR}) with respect to $\bar{Q}$ then is  
    \begin{eqnarray}\label{k1k2}
        K_1 \pdv{R_{\max}}{\bar{Q}} = K_2,
    \end{eqnarray}
    where
    \begin{eqnarray} \small
        K_1 &=& (\frac{\phi}{K}\bar{Q}^2+\frac{\alpha}{K}\bar{Q}^3-1)3R_{\max}^2+3\bar{Q}^2, \label{k111}\\
        K_2 &=& 6\bar{Q}^2 - 6\bar{Q}R_{\max} - 2\frac{\phi}{K}\bar{Q}R_{\max}^3 - 3\frac{\alpha}{K}\bar{Q}^2R_{\max}^3. \label{k222}
    \end{eqnarray}
    
    The rest is to show that both $K_1 > 0$ and $K_2 > 0$. As $R_{max} = (1-x^*)\bar{Q}$, $\bar{Q} \ge R_{\max} >0$. By Eq.~(\ref{k111}), we have 
    \begin{eqnarray*}
        K_1 > -3R_{\max}^2 + 3\bar{Q}^2 \ge 0.
    \end{eqnarray*}
    Meanwhile, by Eq.~(\ref{eqMR}), we have
    \begin{eqnarray} \label{eqK2}
         2\bar{Q}^2 - 3 \bar{Q}R_{\max} + \frac{R_{\max}^3}{\bar{Q}}= (\frac{\phi}{K}\bar{Q} + \frac{\alpha}{K}\bar{Q}^2)R_{\max}^3.
    \end{eqnarray}
    Following Eq.~(\ref{k222}) and Eq.~(\ref{eqK2}), the following inequality holds. 
    \begin{eqnarray*} \small
        K_2 &>& 6\bar{Q}^2 - 6\bar{Q}R_{\max} - 3(\frac{\phi}{K}\bar{Q}R_{\max}^3 + \frac{\alpha}{K}\bar{Q}^2R_{\max}^3) \\
        & = & 3\bar{Q}R_{\max} - 3\frac{R_{\max}^3}{\bar{Q}} \\
        & = & 3R_{\max}\frac{\bar{Q}^2 - R_{\max}^2}{\bar{Q}} \ge 0
    \end{eqnarray*}
    Having both $K_1>0$ and $K_2>0$, by Eq~(\ref{k1k2}), we have 
    \begin{eqnarray*}
        \pdv{R_{\max}}{\bar{Q}} > 0.
    \end{eqnarray*}
    By Theorem~\ref{theoremOPTP1}, the overall delay of the execution of task $A(L, \tau_d)$ is equal to $R_{\max}$, the overall delay thus is a monotonically increasing function of $\bar{Q}$.
\end{proof}

Lemma~\ref{lemmaOPTP1M} and Lemma~\ref{lemmaRM} indicate that a smaller $\bar{Q}$ will result in a less overall delay and cost on the execution of task $A(L, \tau_d)$. In the following we aim to determine the range of $\bar{Q}$.

We first determine an upper bound on $\bar{Q}$ to meet the overall delay requirement of task $A(L, \tau_d)$. We have $(2 - \frac{\alpha}{K}\tau_d^3){\bar{Q}}^3 - (3\tau_d+\frac{\phi}{K}\tau_d^3){\bar{Q}}^2 + \tau_d^3 = 0$ by replacing $R_{\max}$ in Eq.~(\ref{eqMR}) with $\tau_d$, and let $\bar{Q}^*$ be the value of $\bar{Q}$ to ensure the equality holds. By Lemma~\ref{lemmaRM} and the property of the inverse function, $\bar{Q}$ is a monotonically increasing function of $R_{\max}$. Since $R_{\max} \le \tau_d$, we have $\bar{Q} \le \bar{Q}^*$.

We then determine another upper bound on $\bar{Q}$ to meet the CPU-cycle frequency constraint.
By Theorem~\ref{theoremOPTP1}, we have
\begin{eqnarray}
    f_w = \frac{B_0x_0^*}{\bar{Q}(1-x_0^*)} \le f_{\max}. \label{eqfwMf}
\end{eqnarray}

By Eq.~(\ref{eqyx}) and $\xi = y^*/\bar{Q}$ in Lemma~\ref{lemmaOPTP1M}, Inequality~(\ref{eqfwMf}) can be rewritten as $B_0\xi \le f_{\max}$. Thus, 
    $\xi \le \frac{f_{\max}}{B_0}$.

Eq.~(\ref{eqMd}) can be rewritten by replacing $\xi$ with $f_{\max}/B_0$ as follows.
\begin{eqnarray}
    \bar{Q} = \frac{3Kf_{\max}^2B_0-\phi B_0^3}{\alpha B_0^3 - 2Kf_{\max}^3}.
\end{eqnarray}
Since $\bar{Q}$ is a monotonically increasing function of $\xi$ by Lemma~\ref{lemmaOPTP1M}, $\bar{Q}$ should be no greater than $\frac{3Kf_{\max}^2B_0-\phi B_0^3}{\alpha B_0^3 - 2Kf_{\max}^3}$. Let $\bar{Q}_{\max} = \frac{3Kf_{\max}^2B_0-\phi B_0^3}{\alpha B_0^3 - 2Kf_{\max}^3}$.
Therefore, the range of $\bar{Q}$ is $\bar{Q} \le \min \{\bar{Q}^*, \bar{Q}_{\max}\}$.

Following its definition, the derivative of $\bar{Q}$ with respect to $q_i$, where $1 \le i \le n$, is
\begin{eqnarray} \label{eqQq}
    \pdv{\bar{Q}}{q_i} = (\frac{1}{q_i^2})(\sum\limits^{n}_{i=1}\frac{1}{q_i})^{-1} > 0.
\end{eqnarray}

Eq.~(\ref{eqQq}) shows that $\bar{Q}$ is an increasing function of $q_i$, which implies that a smaller $q_i$ will correspond a smaller $\bar{Q}$ for each $i$ with $1\leq i \leq n$.  Thus,  to minimize $\bar{Q}$, we should choose the first smallest $n$ servers from the $N$ servers in terms of the value of their $q_i$.

Let $j_1, j_2, \cdots, j_N$ be the index sequence of servers by sorting their $q)i$ in increasing order, i.e., $q_{j_t} \le q_{j_{t+1}}$ with $1 \le t \le N-1$.
Denote by $Q(n)$ for any $n$ with $1\leq n\leq m\leq N$, which is defined as follows. 
\begin{eqnarray}\label{qn001}
    Q(n) = q_0 + (\sum\limits^{n}_{i=1}\frac{1}{q_{j_i}})^{-1}.
\end{eqnarray}
Recall that the number of servers selected to serve task $A(L, \tau_d)$ cannot exceed $m$, i.e., $n \le m$.
Since $\sum^{n}_{i=1}\frac{1}{q_{j_i}} \le \sum^{m}_{i=1}\frac{1}{q_{j_i}}$, we have
\begin{eqnarray} \label{eqMmin}
    Q(m) \le Q(n), ~~~1 \le n \le m.
\end{eqnarray}

Function $Q(n)$ in Inequality~(\ref{eqMmin}) reaches the minimum when $n = m$. Therefore, the number of servers selected to serve task $A(L, \tau_d)$ should be $m$ to achieve the minimum overall cost  $OPT(P_1)$ of its execution by Lemma~\ref{lemmaOPTP1M}. Therefore, we have the following lemma.

\begin{lemma} \label{lemmasubset}
Let  $s_0=\{s_{j_1}, \cdots, s_{j_m}\}$ be the set of servers allocated to task $A(L,\tau_d)$ in problem $P_1$ to achieve the minimum cost $OPT(P_1^{\mathscr{S}_{s_0}})$.  If $Q(m) \le \min\{\bar{Q}^*, {Q}_{\max}\}$, then  
    \begin{eqnarray} \label{eqoptp1min}
        OPT(P_1^{\mathscr{S}_{s_0}}) = 3K\frac{(y_m^*)^2}{Q(m)^2} - \phi,
    \end{eqnarray}
    where
    $y_m^*$ is the value of variable $y_m$ in the following equation.
    \begin{eqnarray} \label{eqyMm}
        2{y_m}^3 + 3{y_m}^2 = \frac{(\phi+\alpha Q(m))Q(m)^2}{K}.
    \end{eqnarray}
\end{lemma}

\begin{proof}
 The subset of servers $\{s_{j_1}, \cdots, s_{j_m}\}$ can be identified by Ineq.~(\ref{eqMmin}).  By Theorem~\ref{theoremOPTP1}, Lemma~\ref{lemmasubset} then follows.
\end{proof}


\subsection{Algorithm}

Following the discussion in the very beginning in this section, the optimal solution $OPT(P_0)$ to problem $P_0$ is the minimum one between $OPT_{local}$ and $OPT(P_1^{\mathscr{S}_{s_0}}) +E_t+\phi$, while $OPT_{local}$ can be expressed as
\begin{eqnarray}
     \textstyle OPT_{local} = \min \{\kappa \sum^{B_0}_{w=1} (f_w)^2 + \alpha \sum^{B_0}_{w=1} (f_w)^{-1}\}.
\end{eqnarray}
By Eq.~(\ref{eqFfw}) and Eq.~(\ref{eqfm}), 
\begin{eqnarray}\label{eqOPTLOCAL}
    OPT_{local} = B_0(\kappa \bar{f}^2 + \alpha \bar{f}^{-1}).
\end{eqnarray}
Therefore, we have
\begin{eqnarray} \label{eqOPTP00}
    OPT(P_0) = \min \{\frac{3K{y_m^*}^2}{Q(m)^2} + E_t,\  B_0(\kappa \bar{f}^2 + \alpha \bar{f}^{-1})\}.
\end{eqnarray}

The detailed algorithm for problem $P_0$ thus is given in {\tt Algorithm}~\ref{alg01}.

\begin{algorithm}[t!] \footnotesize
    \caption{\tt Task Offloading Scheduling (TOS)}
    \label{alg01}
    \begin{algorithmic}[1]
    \REQUIRE $N$, $m$ and $q_i$ with $1\leq i \leq N$.
    \ENSURE $x_i$ and $f_w$.
    \STATE $x_i \leftarrow 0$, where $i = 0, 1 ,\cdots, N$;
    \STATE A sequence of servers ${j_1}, {j_2}, \cdots, {j_N}$ is obtained by sorting their $q_i$ in increasing order;
    \STATE Calculate $Q(m)$ and $y_m^*$ by Eq.~(\ref{qn001}) and~(\ref{eqyMm}), respectively;
    \STATE Calculate $OPT(P_1^{\mathscr{S}_{s_0}})$ by Eq.~(\ref{eqoptp1min});
    \STATE Calculate $OPT_{local}$ by Eq.~(\ref{eqOPTLOCAL});
    \STATE Calculate $OPT(P_0)$ by Eq.~(\ref{eqOPTP00});
    \IF{$OPT(P_0) = OPT_{local}$}
    \RETURN $x_0 \gets 1$ and $f_w \gets \bar{f}$, EXIT.
    \ELSE
    \STATE Calculate $x_0^*$ by Eq.~(\ref{eqxy});
    \STATE $Q \gets \sum^{m}_{i=1} {1}/{q_{j_i}}$;
    \STATE Calculate $x_{j_k}$ for all $k$ with $1\leq k \leq m$ by Eq.~(\ref{eqxiQ});
    \STATE Calculate $f_w$ by Eq.~(\ref{eqfwM});
    \RETURN $x_{j_k}$ for all $k$ with $1\leq k \leq m$, and $f_w$.
    \ENDIF
    \end{algorithmic}
\end{algorithm}

\setlength{\belowcaptionskip}{-15pt}


\begin{theorem}
{\tt Algorithm}~\ref{alg01} for the delay-energy joint optimization problem in an MEC delivers an optimal solution if $Q(m) \le \min\{\bar{Q}^*, {Q}_{\max}\}$, and its time complexity is $O(N\log N)$, where $N$ is the number of servers in the MEC.
\end{theorem}

\begin{proof}
    By Theorem~\ref{theoremOPTP1} and Lemma~\ref{lemmasubset}, Algorithm~\ref{alg01} can find an optimal solution to the delay-energy joint optimization problem if $Q(m) \le \min\{\bar{Q}^*, {Q}_{\max}\}$. Its time complexity analysis is quite obvious, which is $O(N \log N)$ due to sorting.
\end{proof}

\section{Performance Evaluation}\label{sec05}
In this section, we evaluate the performance of the proposed algorithm for the delay-energy joint optimization problem through simulations. We also investigate the impacts on important parameters on the performance of the proposed algorithm.

\subsection{Experimental settings}
We set the number of servers in an MEC at $N = 100$, the data transmission rate between the AP and a server $r_i$ is uniformly distributed between 100~$Mbps$ and 1~$Gbps$. The computing capacity of a server $c_i$ is uniformly distributed between 1~GHz and 4~GHz. The transmission rate between the mobile device and the AP is set at $r_{h,p} = 2.5$~Mbps and the transmission power of the mobile device $P_{tx} = 0.5$~W. Besides, $E_t = 0.15$~J, $\kappa = 10^{-26}$, $f_{\max} = 2$~GHz, and the size of task $A(L,\tau_d)$ is $L = 50$~KB~\cite{MZL16}. Assume that the computation intensity of the task is 700 cycles per bit~\cite{MN10}. Under these experimental settings, $Q(m) \le \min\{\bar{Q}^*, Q_{\max}\}$ holds.

To evaluate the proposed algorithm, {\tt Algorithm} 1,  three benchmark algorithms are introduced as follows. 

 \begin{itemize}
     \item \textbf{Local\_Execution:} The entire task is executed at the mobile device and the CPU-cycle frequencies at the mobile device are scheduled to minimize the overall cost. The optimal value of this policy is $OPT_{local}$.

    \item \textbf{MEC\_Execution:} The entire task is offloaded to MEC. Then the subset of servers to serve the task and the task allocation among the selected servers are both optimal. The difference between MEC Execution and TOS is that $x_0 = 0$ in MEC Execution policy.

    \item \textbf{Mixed\_Execution:} The task is executed in both the mobile device and MEC. The difference between algorithms {\tt Mixed\_Execution} and {\tt TOS} is in the setting of $x_0 = 1 / (1 +m)$ in algorithm {\tt Mixed\_Execution}.
     \end{itemize}

\begin{figure*}[t!]
    \centering
  \subfloat[The overall cost on different $m$]{%
       \includegraphics[width=0.4\linewidth]{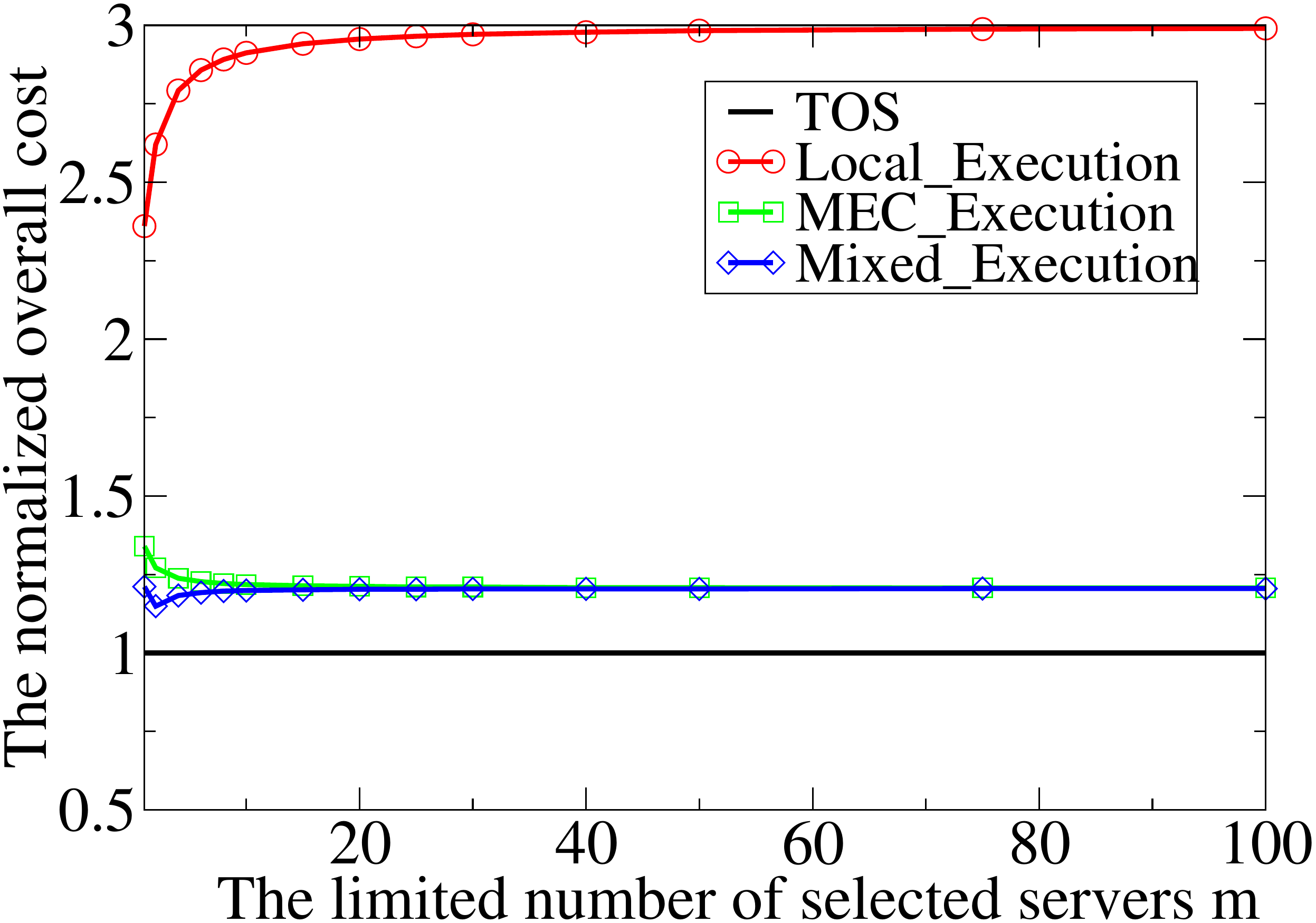}} \label{fig:weightoverallm}
  \quad\quad\quad\quad\subfloat[The overall cost on different weight $\alpha$]{%
       \includegraphics[width=0.4\linewidth]{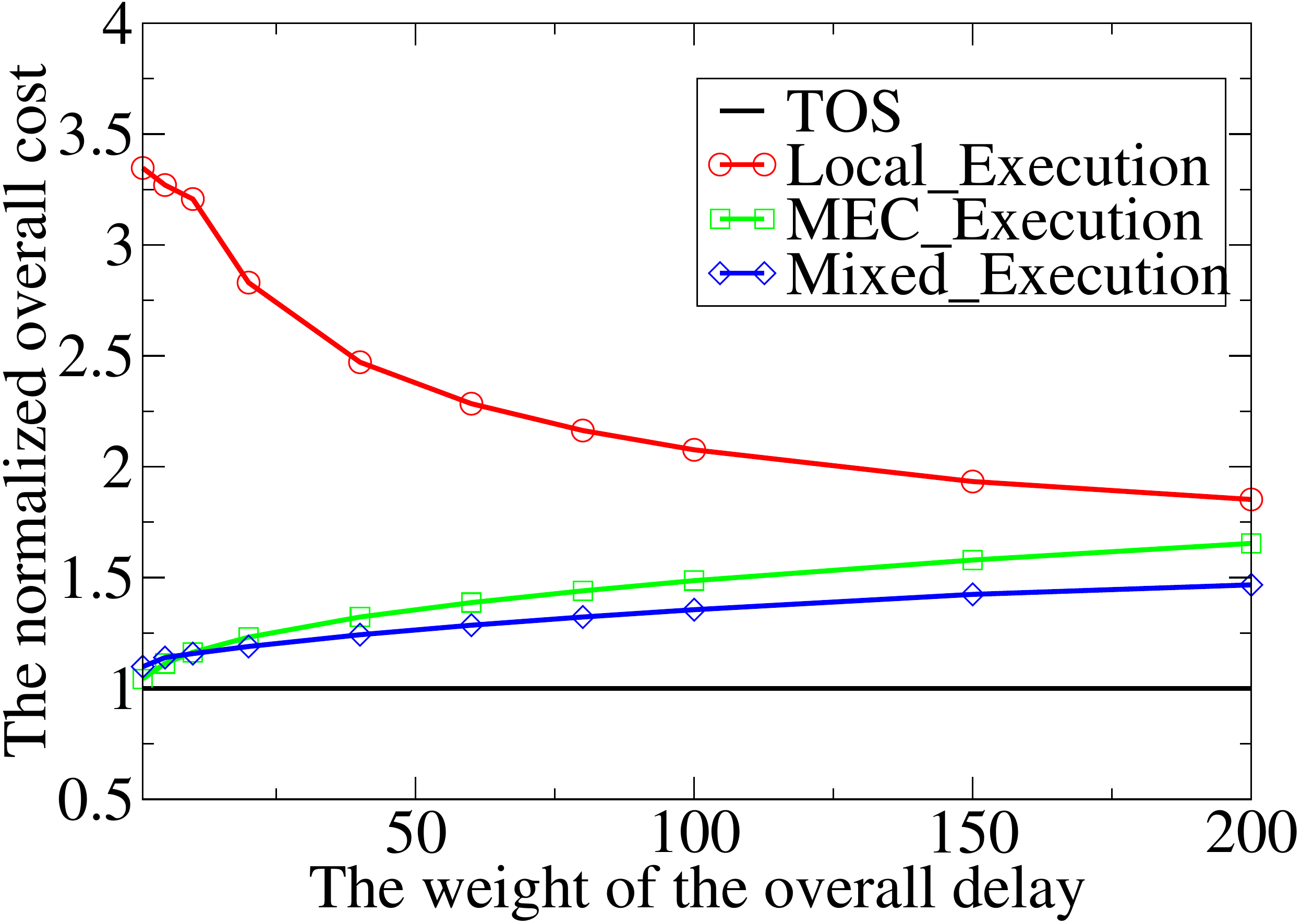}} \label{fig:weightoveralla}
\vspace{-2mm}
\caption{The performance of algorithms {\tt TOS}, {\tt Local \_Execution}, {\tt MEC \_Execution} and {\tt Mixed \_Execution}.}
\vspace{4mm}
\label{fig:performance}
\end{figure*}

\begin{figure*}[t!]
    \centering
  \quad\subfloat[Impact of $m$ on the overall delay]{%
        \includegraphics[width=0.42\linewidth]{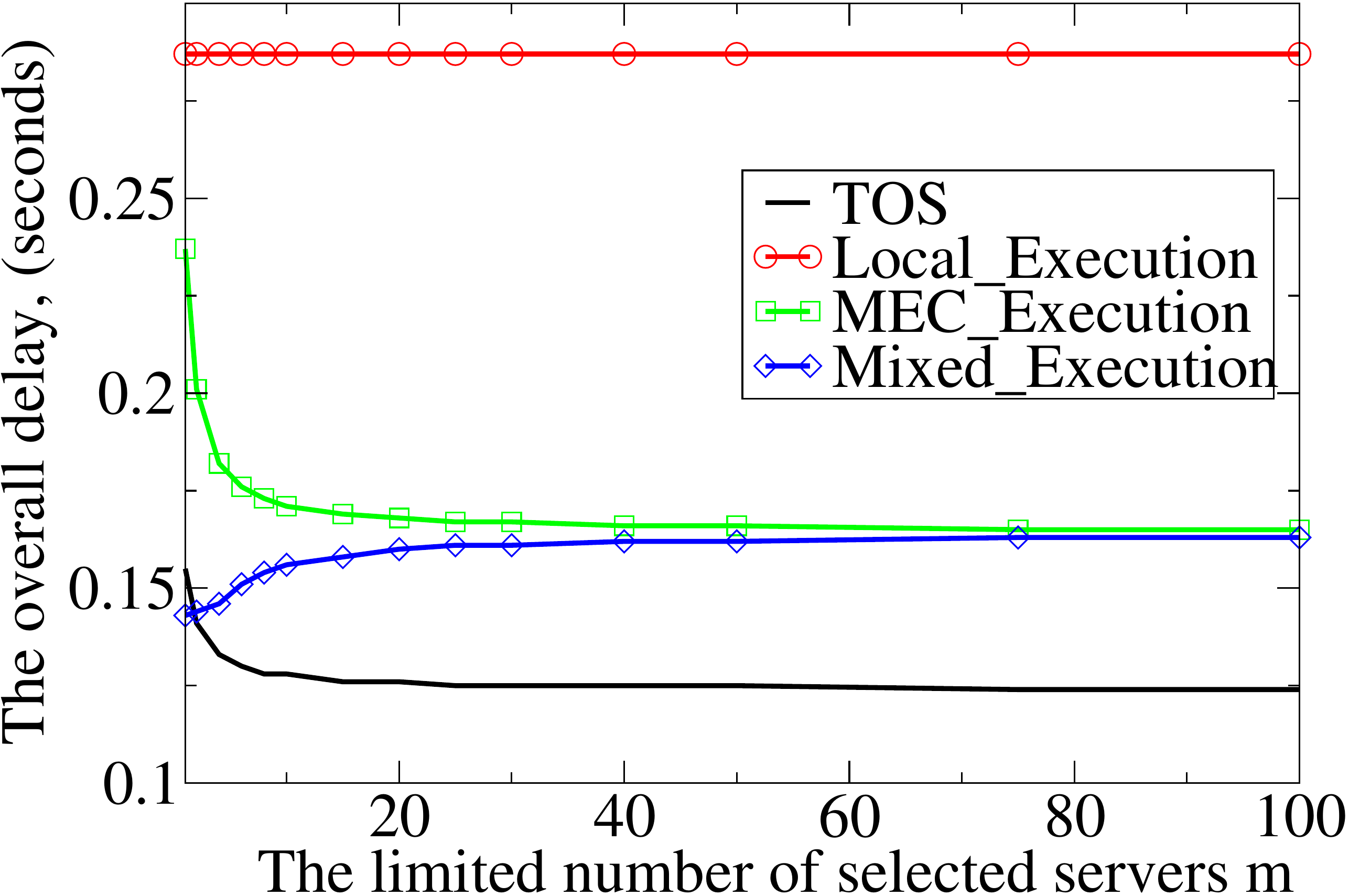}}
    \label{fig: limitdelay}\quad\quad\quad\quad
  \subfloat[Impact of $m$ on the energy consumption]{%
        \includegraphics[width=0.4\linewidth]{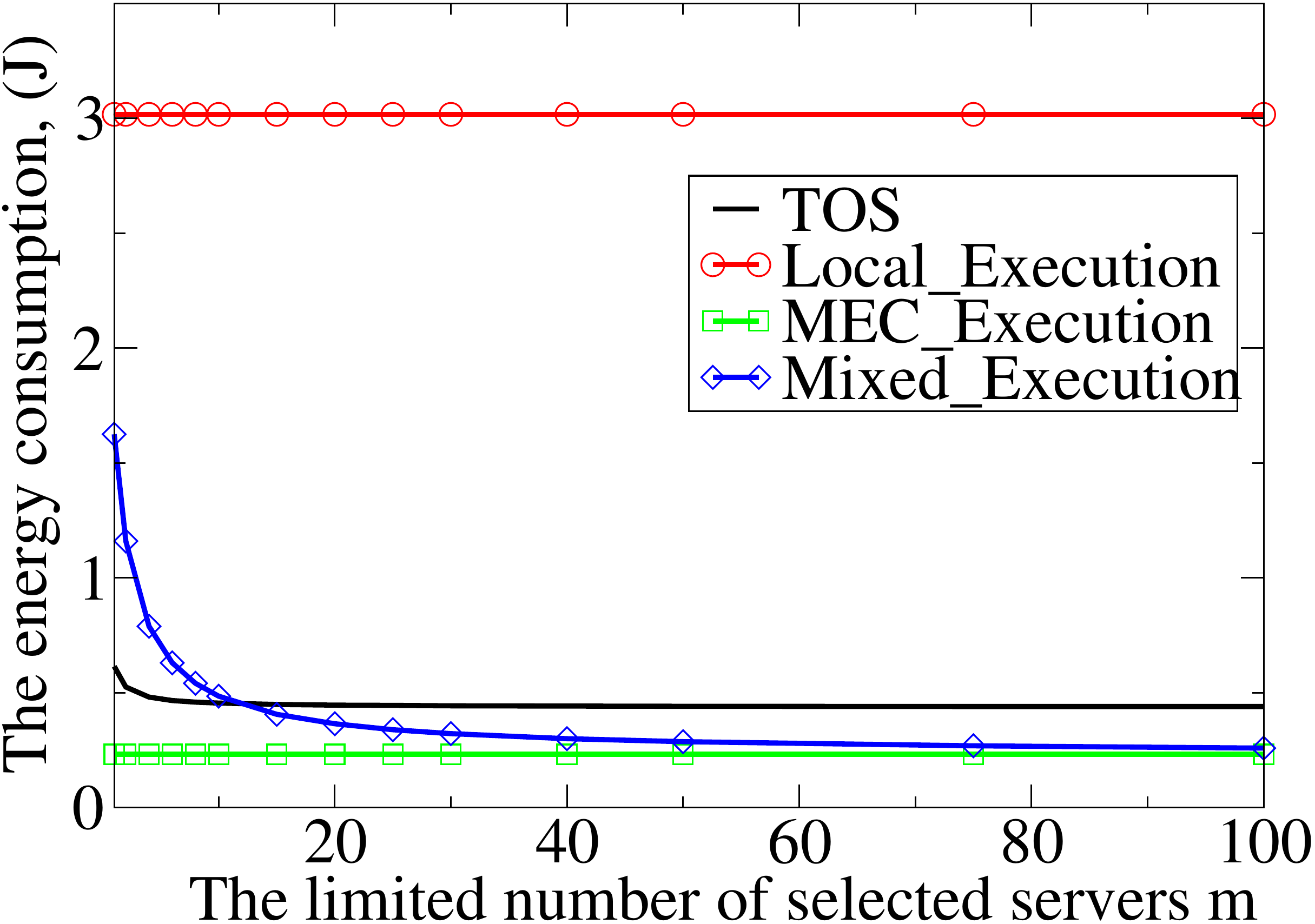}}
    \label{fig: limitenergy}\qquad
    \vspace{-2mm}
  \caption{The impact of $m$ with $\alpha = 20$ on the overall delay and energy consumption.}
  \label{fig: limit}
  \vspace{3mm}
\end{figure*}

\begin{figure*}[t!]
    \centering
  \subfloat[Impact of weight $\alpha$ on the overall delay]{%
        \includegraphics[width=0.4\linewidth]{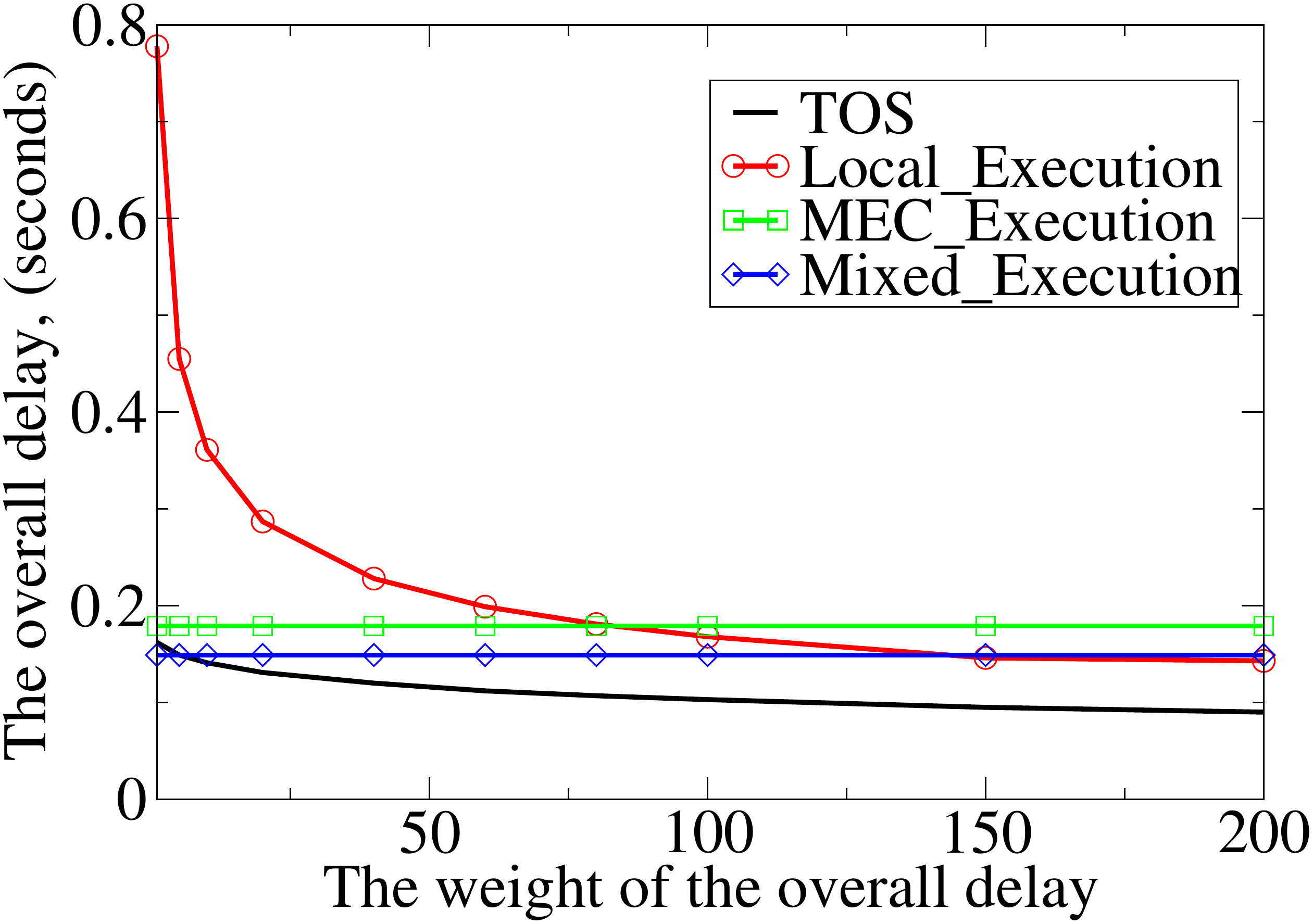}} \label{fig: weightdelay}
  \qquad\quad\quad\subfloat[Impact of weight $\alpha$ on the energy consumption]{%
        \includegraphics[width=0.4\linewidth]{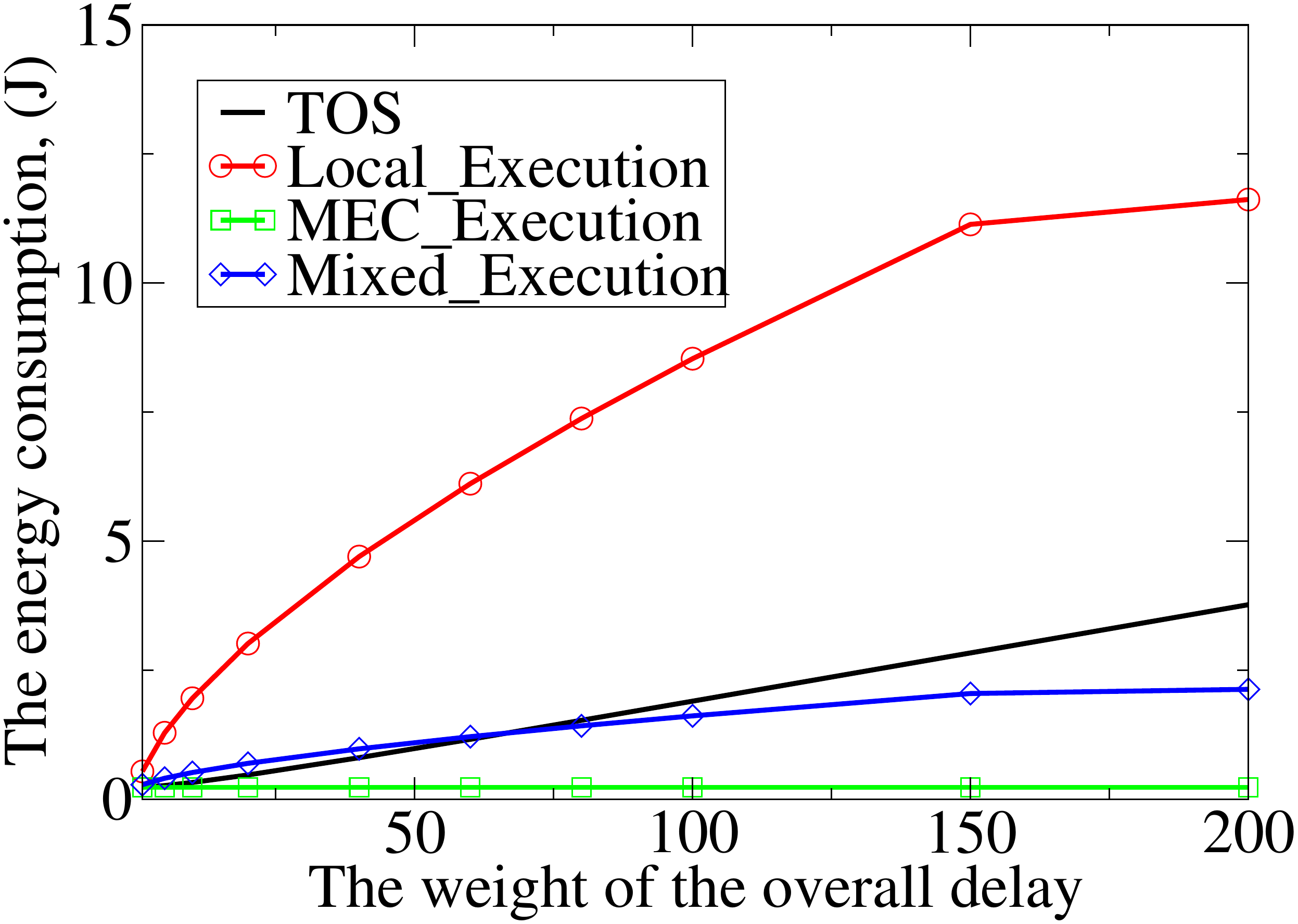}} \label{fig: weightenergy}
    \vspace{-2mm}
  \caption{The impact of weight $\alpha$ with $m =5$ on the overall delay and energy consumption.}
  \label{fig: weight}
  \vspace{4mm}
\end{figure*}

\subsection{Performance evaluation}
We first evaluate the performance of algorithm {\tt TOS} against benchmark algorithms {\tt Local\_Execution}, {\tt MEC\_Execution}, and {\tt Mixed\_Execution}, by varying both numbers of servers $m$ and $\alpha$. The overall costs delivered by different algorithms are normalized based on the overall cost delivered by {\tt TOS}. 

Fig.~\ref{fig:performance} shows the overall costs of task execution delivered by the mentioned algorithms. We can see from Fig.~\ref{fig:performance}~(a) that the overall cost by algorithm {\tt TOS} is around 80\% of those by algorithms {\tt MEC\_Execution} and {\tt Mixed\_Execution}, and 34\% of that by algorithm {\tt Local\_Execution} when $\alpha= 20$. In Fig.~\ref{fig:performance}~(b), $m = 5$.
Similarly, we can see from Fig.~\ref{fig:performance}~(b) that algorithm {\tt TOS} outperforms the other algorithms.

\subsection{Impact of different parameters}
We then study the impacts of important parameters: $m$ and $\alpha$ on the overall delay and energy consumption of the proposed algorithm {\tt TOS}.

We investigate the impact of the number of chosen servers $m$ on the performance of algorithms {\tt TOS}, {\tt Local\_Execution}, {\tt MEC\_Execution}, and {\tt Mixed\_Execution} in terms of the overall delay and the energy consumption on the mobile device. We vary $m$ from 1 to 200 while fixing $\alpha$ at 20. It can be seen from Fig.~\ref{fig: limit}~(a) that algorithm {\tt TOS} has the minimum overall delay because of its optimal task allocation among the mobile device and the selected servers. Particularly, when $m\leq 8$, it greatly impacts the overall delay by the mentioned  algorithms; otherwise, its impact is negligible.

Fig.~\ref{fig: limit}~(b) shows that algorithm {\tt TOS} is not energy-efficient as it requires local processing at the mobile device to minimize the overall delay when $m\leq 8$; otherwise, its energy consumption impact is negligible.

We also evaluate the impact of parameter $\alpha$ on the performance of the mentioned algorithms in terms of the overall delay and energy consumption. It can be seen from Fig.~\ref{fig: weight}~(a) that the increase on the value of $\alpha$ will shorten the overall delay of the task execution when $\alpha \leq 70$; otherwise, its impact on overall delay is negligible. However, as shown in Fig.~\ref{fig: weight}~(b), algorithms {\tt TOS} and {\tt Local\_Execution} will incur more energy consumption at the mobile device. 


\section{Related Work}\label{sec06}
MEC has attracted lots of attention from both industries and academia, see~\cite{MYZH17} for a comprehensive survey.

Many previous studies on MEC focused on designing a system to support task offloading and allocation~\cite{RRPK98,CBCW+10,KAHM+12,LMZL16,YHCK17,ZWGK+13,SSB15,MZL16}. Several offloading frameworks including MAUI~\cite{CBCW+10} and ThinkAir~\cite{KAHM+12}, were proposed to prolong the battery lifetime of mobile devices and reduce processing delays of computation tasks. Particularly, Liu {\it et al.}~\cite{LMZL16} introduced a delay-optimal task scheduling policy with random task arrivals for a single-user MEC system. For multi-user MEC systems, Chen {\it et al.}~\cite{CJLF16} proposed a distributed offloading algorithm which can achieve the Nash equilibrium to reduce wireless interferences. You {\it et al.}~\cite{YHCK17} designed a centralized task offloading system for MEC based on time-division multiple access and orthogonal frequency-division multiple access. All of these mentioned work assumed non-adjustable processing capabilities of CPUs at mobile devices. However, it is energy inefficient as the energy consumption of a CPU increases super-linearly with its CPU-cycle frequency~\cite{BB96}.

Several recent research has applied the dynamic voltage and frequency scaling (DVFS) technique into task offloading in an MEC environment. For example,  Zhang {\it et al.}~\cite{ZWGK+13} proposed an energy-optimal binary offloading policy, i.e., either executing a task on a mobile device or offloading the entire task to MEC, under stochastic wireless channels for a single-user MEC system with the DVFS technique. In addition, Sardellitti {\it et al.}~\cite{SSB15} proposed an algorithm for both communication and computational resource allocations in multi-cell MIMO cloud computing systems. Mao {\it et al.}~\cite{MZL16} studied an MEC system with energy harvesting devices and focused on reducing the energy consumption of mobile devices by scheduling CPU-cycle frequencies.


\section{Conclusion}\label{sec07}
In this paper, we studied task offloading from mobile devices to MEC. We formulated a novel delay-energy joint optimization problem. 
We first formulated the problem as a mixed-integer nonlinear program problem. We then performed a relaxation to relax the original problem to a nonlinear programming problem whose solution can be found in polynomial time, and showed how to make use of the solution to the relaxed problem to find a feasible solution to the problem of concern in this paper. We finally evaluate the performance of the proposed algorithm thorough experimental simulations. Experimental results demonstrated that the proposed algorithm outperforms the three baselines.

\section*{APPENDIX}
\subsection*{A. Proof of Inequalities~(\ref{eqjensen})}
Define $\varphi(x) = \frac{1}{x^2}$ with $x > 0$, which is a convex function. Let $\lambda_w = \frac{1}{B}$, where $1 \le w \le B$. By Jensen's Inequality, we have
\begin{eqnarray} \label{eqjensen2}
    \varphi(\sum_{w=1}^{B}\lambda_w x_w) \le \sum_{w=1}^{B}\lambda_w \varphi(x_w)
\end{eqnarray}
Let $x_w = \frac{1}{f_w}$, where $1 \le w \le B$. Recall $R' = \sum_{w=1}^{B}\frac{1}{f_w}$. The left side of Inequality~(\ref{eqjensen2}) is $\varphi(\sum_{w=1}^{B}\lambda_w x_w) = \varphi(R'/B) = B^2/R'^2$ and the right side is $\sum_{w=1}^{B}\lambda_w \varphi(x_w) = \frac{1}{B}\sum_{w=1}^{B}f_w^2$. Hence, by Inequality~(\ref{eqjensen2}) we have $\sum_{w=1}^{B}f_w^2 \ge \frac{B^3}{R'^2}$. The equality holds when $f_w = B/R'$. 
Recall $R' \le R_{max}$. Inequalities~(\ref{eqjensen}) thus holds.


\begin{thebibliography}{10}

\bibitem{AA16}
A. Ahmed and E. Ahmed.
\newblock A survey on mobile edge computing.
\newblock {\it Intelligent Systems and Control (ISCO), 2016 10th International Conference on}, pp. 1--8, IEEE, 2016.

\bibitem{ABBC+17}
G. Ananthanarayanan, P. Bahl, P. Bod{\'\i}k, K. Chintalapudi, M. Philipose, L. Ravindranath, and S. Sinha.
\newblock Real-time video analytics: The killer app for edge computing.
\newblock {\it Computer}, Vol. 50, no. 10, pp. 58--67, 2017.

\bibitem{BBCH17}
T. Braud, F. H. Bijarbooneh, D. Chatzopoulos, and P. Hui.
\newblock Future networking challenges: The case of mobile augmented reality.
\newblock {\it Distributed Computing Systems (ICDCS), 2017 IEEE 37th International Conference on}, pp. 1796--1807, IEEE, 2017.

\bibitem{BB96}
T. D. Burd and R. W. Brodersen.
\newblock Processor design for portable systems.
\newblock {\em Technologies for wireless computing}, pp. 119--137, Springer, 1996.

\bibitem{CJLF16}
X. Chen, L. Jiao, W. Li, and X. Fu.
\newblock Efficient multi-user computation
offloading for mobile-edge cloud computing.
\newblock {\em IEEE/ACM Transactions on Networking}, Vol. 24, no. 5, pp. 2795--2808, 2016.

\bibitem{CSMM+16}
J. Cho, K. Sundaresan, R. Mahindra, J. E. van der Merwe, and S. Rangarajan,
\newblock Acacia: Context-aware edge computing for continuous interactive applications over mobile networks.
\newblock {\em Proc. of MobiCom}, pp. 505--506, 2016.

\bibitem{CBCW+10}
E. Cuervo, A. Balasubramanian, D.-k. Cho, A. Wolman, S. Saroiu, R. Chandra, and P. Bahl.
\newblock Maui: making smartphones last longer with code offload.
\newblock {\em Proceedings of the 8th international conference on Mobile systems, applications, and services}, pp. 49--62, ACM, 2010.


\bibitem{HC17}
W. Hu and G. Cao.
\newblock Quality-aware traffic offloading in wireless networks.
\newblock {\em IEEE Transactions on Mobile Computing}, 2017.

\bibitem{HPSS+15}
Y. C. Hu, M. Patel, D. Sabella, N. Sprecher, and V. Young.
\newblock Mobile edge computing {--} a key technology towards 5g.
\newblock {\it ETSI White Paper}, Vol. 11, no. 11, pp. 1--16, 2015.

\bibitem{KAHM+12}
S. Kosta, A. Aucinas, P. Hui, R. Mortier, and X. Zhang.
\newblock Thinkair: Dynamic resource allocation and parallel execution in the cloud for mobile code offloading.
\newblock {\em Infocom, 2012 Proceedings IEEE}, pp. 945--953, IEEE, 2012.

\bibitem{LCLV17}
Y. Li, Y. Chen, T. Lan, and G. Venkataramani.
\newblock Mobiqor: Pushing the envelope. of mobile edge computing via quality-of-result optimization.
\newblock {\em Distributed Computing Systems (ICDCS), 2017 IEEE 37th International Conference on}, pp. 1261--1270, IEEE, 2017.

\bibitem{LMZL16}
J. Liu, Y. Mao, J. Zhang, and K. B. Letaief.
\newblock Delay-optimal computation task scheduling for mobile-edge computing systems.
\newblock {\em Information Theory (ISIT), 2016 IEEE International Symposium on}, pp. 1451--1455, IEEE, 2016.

\bibitem{LB17}
P. Mach and Z. Becvar.
\newblock Mobile edge computing: A survey on architecture and computation offloading.
\newblock {\em IEEE Communications Surveys \& Tutorials}, Vol. 19, pp. 1628--1656, IEEE, 2017.

\bibitem{MYZH17}
Y. Mao, C. You, J. Zhang, K. Huang, and K. B. Letaief.
\newblock A survey on mobile edge computing: The communication perspective.
\newblock {\em IEEE Communications Surveys \& Tutorials}, 2017.

\bibitem{MZL16}
Y. Mao, J. Zhang, and K. B. Letaief.
\newblock Dynamic computation offloading for mobile-edge computing with energy harvesting devices.
\newblock {\em IEEE Journal on Selected Areas in Communications}, Vol. 34, no. 12, pp. 3590--3605, 2016.

\bibitem{MA10}
R. T. Marler and J. S. Arora.
\newblock The weighted sum method for multi-objective optimization: new insights.
\newblock {\em Structural and multidisciplinary optimization}, Vol. 41, no. 6, pp. 853--862, 2010.

\bibitem{MN10}
A. P. Miettinen and J. K. Nurminen.
\newblock Energy efficiency of mobile clients in cloud computing.
\newblock {\em Proc. of HotCloud}, Vol. 10, pp. 4--4, 2010.

\bibitem{RCN02}
J. M. Rabaey, A. P. Chandrakasan, and B. Nikolic.
\newblock {\em Digital integrated circuits}, Vol. 2. Prentice hall Englewood Cliffs, 2002.

\bibitem{RRPK98}
A. Rudenko, P. Reiher, G. J. Popek, and G. H. Kuenning.
\newblock Saving portable computer battery power through remote process execution.
\newblock {\em ACM SIGMOBILE Mobile Computing and Communications Review}, Vol. 2, no. 1, pp. 19--26, 1998.

\bibitem{SSB15}
S. Sardellitti, G. Scutari, and S. Barbarossa.
\newblock Joint optimization of radio and computational resources for multicell mobile-edge computing.
\newblock {\em IEEE Transactions on Signal and Information Processing over Networks}, Vol. 1, no. 2, pp. 89--103, 2015.

\bibitem{SSXP+15}
M. Satyanarayanan, P. Simoens, Y. Xiao, P. Pillai, Z. Chen, K. Ha, W. Hu, and B. Amos.
\newblock Edge analytics in the internet of things.
\newblock {\it IEEE Pervasive Computing}, Vol. 14, no. 2, pp. 24--31, 2015.

\bibitem{YHCK17}
C. You, K. Huang, H. Chae, and B.-H. Kim.
\newblock Energy-efficient resource allocation for mobile-edge computation offloading.
\newblock {\em IEEE Transactions on Wireless Communications}, Vol. 16, no. 3, pp. 1397--1411, 2017.

\bibitem{ZWGK+13}
W. Zhang, Y. Wen, K. Guan, D. Kilper, H. Luo, and D. O. Wu.
\newblock Energy-optimal mobile cloud computing under stochastic wireless channel.
\newblock {\em IEEE Transactions on Wireless Communications}, Vol. 12, no. 9, pp. 4569--4581, 2013.

\end{thebibliography}

\end{document}